\let\ORIlabel\label
\let\ORIrefstepcounter\refstepcounter
\AddToHook{package/hyperref/before}{
    \let\label\ORIlabel
    \let\refstepcounter\ORIrefstepcounter
}
\documentclass[final,onefignum,onetabnum]{siamart220329}

\usepackage[
  left=1in,
  right=1in,
  top=1in,
  bottom=1in
]{geometry}

\usepackage{amsfonts,amsmath,mathrsfs,bm}
\usepackage{amssymb}
\usepackage[notref,notcite]{showkeys}
\usepackage{graphicx}
\usepackage{wrapfig}
\usepackage{subfig}
\usepackage{float}
\usepackage{bbm}
\usepackage{algorithmic, algorithm}
\usepackage{gensymb,color,soul}
\usepackage{tikz}
\usepackage{tikz-3dplot}

\title{Computation of frequency- and time-domain Jacobians in optical tomography with Monte Carlo simulations}

\usepackage{lineno}

\author{
Pauliina Hirvi\footnotemark[2]
\and Jaakko Olkkonen\footnotemark[2] \footnotemark[3]
\and Qianqian Fang\footnotemark[4]
\and Ilkka Nissilä\footnotemark[5]
}

\begin{document}
\maketitle

\renewcommand{\thefootnote}{\fnsymbol{footnote}}
\footnotetext[2]{Aalto University, Department of Mathematics and Systems Analysis, P.O.~Box 11100, FI-00076 Aalto, Finland (pauliina.hirvi@aalto.fi, jaakko.olkkonen@aalto.fi). Pauliina Hirvi and Jaakko Olkkonen contributed equally to this work and are the corresponding authors.}
\footnotetext[3]{Finnish Geospatial Research Institute FGI, National Land Survey of Finland, Department of Remote Sensing and Photogrammetry, P.O.~Box~84, FI-00521 Helsinki, Finland.}
\footnotetext[4]{Northeastern University, Department of Bioengineering, 360 Huntington Avenue, Boston, MA 02115, USA (q.fang@northeastern.edu).}
\footnotetext[5]{Aalto University, Department of Neuroscience and Biomedical Engineering, P.O.~Box 12200, FI-00076 Aalto, Finland (ilkka.nissila@aalto.fi)}

\noindent \begin{abstract}

\noindent \textbf{Significance:} Jacobians, or spatially resolved sensitivity profiles, are central to image reconstruction in model-based optical tomography of biological tissue. Although Monte Carlo (MC) simulations are the gold standard for modeling light transport in turbid media, methodology for frequency- and time-domain Jacobians remains incomplete.\\

\noindent \textbf{Aim:} This work extends MC to directly compute absorption and scattering Jacobians for frequency-domain (amplitude and phase) and time-domain (intensity and mean time-of-flight) measurements and prism-terminated optical fiber detectors.\\

\noindent \textbf{Approach:} Jacobians are derived in the perturbation MC framework and implemented in the high-performance, open-source Monte Carlo eXtreme (MCX) simulator. Results are validated against the diffusion approximation (DA) solved using the finite element method in neonatal head models. MC with split voxels on curved surfaces is extended to Jacobian computation. The detector model is implemented in post-processing and compared with isotropic reception at surface.\\

\noindent\textbf{Results:} MC- and DA-derived Jacobians show excellent agreement only in high-scattering regimes, highlighting the importance of MC for low-scattering domains. The detector model reduces surface sensitivity and marginally increases sensitivity to deeper tissues at short ($< 2$~cm) source--detector separations.\\

\noindent \textbf{Conclusion:} A complete theoretical framework and MC software for computing frequency- and time-domain Jacobians is provided. Realistic detector modeling is encouraged for short-separation channels.

\end{abstract}

\renewcommand{\thefootnote}{\arabic{footnote}}


\pagestyle{myheadings}
\thispagestyle{plain}
\markboth{P.~HIRVI, J.~OLKKONEN, Q.~FANG, AND I.~NISSILÄ}{COMPUTATION OF FREQUENCY- AND TIME-DOMAIN JACOBIANS}

\section{Introduction}
\label{sec:introduction}

Optical tomography (OT) is a class of non-invasive imaging techniques that use visible or near-infrared light to recover the spatial distribution of optical properties inside a scattering medium, such as biological tissue. The objective of OT is to reconstruct three-dimensional maps of absorption and scattering coefficients (absolute or difference imaging) or their dynamic changes relative to a (pre-stimulus) baseline (functional imaging) from measurements of light at the boundary~\cite{arridge1999optical,nissila2005diffuse}. Optical tomography has been applied in a range of biomedical settings, including functional brain imaging, neonatal cerebral monitoring, optical mammography, muscle imaging, and preclinical imaging in small animals~\cite{wheelock2019high, darne2014small, HIELSCHER200579, grosenick2016review,hoshi2016overview, gibson2005recent}. Compared to conventional medical imaging modalities, OT offers safe, non-ionizing, and non-invasive imaging, portable instrumentation of moderate cost, and some tolerance for subject movement~\cite{nissila2005diffuse, wheelock2019high}.  

Optical recordings can be acquired using continuous-wave (CW), time-domain (TD), or frequency-domain (FD) instruments~\cite{nissila2005diffuse}. TD and FD modalities provide time-resolved information that enables the simultaneous recovery of both absorption and scattering properties of the medium~\cite{liebert2003evaluation, McBride:00, yamamoto2016frequency,kangasniemi2024stochastic}, while multiple parameter combinations can produce identical difference data in CW~\cite{arridge1998nonuniqueness, nissila2004}. A TD instrument inputs short light pulses and measures the time-of-flight (TOF) distribution of the detected photons. In FD-OT, light sources are intensity-modulated at radio frequencies, and the amplitude and phase delay of the detected photon density wave are measured~\cite{nissila2002instrumentation,nissila2005instrumentation}. In CW, the modulation frequency is (near) zero and intensity is measured~\cite{nissila2005instrumentation}. TD instruments are generally the most information-rich for imaging, as they measure the full TOF distribution of detected photons. By contrast, FD instruments measure phase shifts that approximately match the mean TOF. By selectively excluding photons with the shortest TOFs, TD measurements can further enhance sensitivity to deeper regions of the medium. However, TD instruments are less widely adopted than CW or FD systems due to their higher cost, greater technical complexity, larger physical size, and limited maximum photon count rate~\cite{nissila2006comparison,ban2022kernel}.

Image reconstruction in OT requires both an accurate model of light propagation and a method to solve the associated severely ill-posed inverse problem. In functional OT of the human brain or other applications where low-contrast changes are expected, the problem is typically linearized by using the first order Taylor approximation to the relationship between measured difference data and changes in the optical parameters. Linearization is obtained via the Fr\'{e}chet derivative of the forward model and its discrete representation as the Jacobian matrix, also referred to as the sensitivity profiles of the measurements. In practice, incorporating multiple measurement types or modes with complementary sensitivity profiles provides more information for the linearized inverse problem, potentially improving reconstruction accuracy~\cite{doulgerakis2019high,hirvi2023improved}. For instance, the temporal point spread function (TPSF) can be recovered from TD measurements by deconvolution and then transformed into amplitude and phase data via a Fourier transform, enabling reconstruction similar to the FD case~\cite{nissila2006comparison, jiang2020image}. Alternatively, the TPSF can be expressed in terms of mean TOF or higher-order moments, which are independent of any fixed modulation frequency~\cite{wabnitz2020depth, ban2022kernel, hillman2000calibration}. Additional FD data types can also be derived mathematically~\cite{sassaroli2023novel}.

The forward problem in OT is to compute the radiance and resulting measurements, given the system geometry, optical properties, and source--detector configuration. The radiative transfer equation (RTE) provides a physically accurate model of light transport in, e.g., biological tissue, but its analytical solutions are limited, and numerical solution remains challenging. As an integro-differential equation, the RTE can lead to prohibitively large numerical problems unless simplifying approximations are employed~\cite{heino2002estimation}. 

Monte Carlo (MC) simulations provide a stochastic approach to modeling photon propagation and are regarded as the gold standard for numerically solving the RTE~\cite{fang2009monte}. MC methods implement the RTE by sampling random photon trajectories from the light sources according to statistical rules inferred from the RTE. A comprehensive overview of the principles of MC simulations of light transport can be found in~\cite{prahl1989monte, sassaroli2012equivalence}. Briefly, each trajectory consists of a sequence of straight-line segments separated by collision points, where the scattering direction is sampled from the scattering phase function. For enhanced computational efficiency, a variance reduction technique called implicit capture is often implemented, where photon packets are propagated instead of single photons~\cite{prahl1989monte}. If the packet is detected, the final weight of the packet describes the probability that a photon would take that path without being absorbed. Detector measurements are estimated from the weights assigned to the photon packets that reach the detector. Several variants of the MC method have been developed, including the albedo-weight method, albedo-rejection method, microscopic Beer--Lambert law (mBLL) method, and absorption-scattering path length rejection method~\cite{sassaroli2012equivalence}. 

MC methods have been used to simulate Jacobians and perform reconstructions in OT in a range of experimental and methodological settings. Absorption changes have been reconstructed from FD amplitude data in studies of affective touch~\cite{jonsson2018affective,maria2022imaging,hirvi2023improved,shekhar2024maternal}, emotional speech~\cite{shekhar2019hemodynamic,maria2020relationship}, and other task or stimulus-related brain activation paradigms~\cite{kotilahti2006application, nasi2013effect, autti2025simultaneously}. FD amplitude or TD intensity have also been combined with FD phase or TD mean TOF, respectively, to reconstruct absorption changes from measured~\cite{hirvi2023improved} and simulated difference data~\cite{heiskala2007optical, heiskala2009probabilistic}. Beyond functional brain imaging, MC reconstructions have been applied to whole-body small-animal imaging using CW and time-gated measurements~\cite{chen2009time}. Absolute imaging studies have addressed the simultaneous reconstruction of absorption and scattering coefficients in two-dimensional settings using simulated FD data~\cite{yamamoto2016frequency,kangasniemi2024stochastic}.  

Although MC methods can be computationally intensive and subject to random sampling noise, advances in graphics processing unit (GPU) technology have enabled massively parallel simulations of large numbers of photon packets within practical time frames~\cite{fang2009monte}. In this work, we use Monte Carlo eXtreme (MCX), which is a GPU-accelerated, NVIDIA CUDA-based light transport simulator widely adopted by the biophotonics community~\cite{fang2009monte, mcx_space, mcx_github}. MCX enables simulating even trillions of photons on modern graphics hardware to enhance signal-to-noise ratios (SNR) in the Jacobians for complex 3D geometries. The Jacobians are computed based on the perturbation MC (pMC) method in a two-step approach, known as the ``replay'' mode, where the baseline simulation selects the photon packets that are detected and the second simulation replays these paths while collecting sensitivity values~\cite{yao2018direct}. MCX implements the mBLL method~\cite{sassaroli2012equivalence}, thus the presentation that follows is based on the mBLL approach. 

Perturbation-based formulas for CW and FD Jacobians with respect to the absorption coefficient ($\mu_a$) and the scattering coefficient ($\mu_s$) have previously been published~\cite{heiskala2007optical, heiskala2009accurate, yamamoto2016frequency, yao2018direct, hirvi2023effects, amendola2024_part1, kangasniemi2024stochastic}, but methods for computing TD Jacobians remain incomplete to the best of our knowledge. The present work carefully reviews the theory of FD and TD Jacobians starting from the principles of light propagation in mBLL, and implements the FD scattering Jacobians as well as TD absorption and scattering Jacobians in MCX. The MC Jacobians are compared with those obtained using the diffusion approximation (DA) of the RTE, which is solved using the finite element method (FEM) with an in-house solver developed at Aalto University. Other available FEM-based DA solvers include, for example, Toast++~\cite{schweiger2014toast++}, the NIRFAST family~\cite{dehghani2009near, cao2025nirfasterff}, and Redbird~\cite{fang2004, fang2009combined, redbird_github}. Comparisons between DA solved with FEM and RTE solved with MCX have been reported previously; for example, the earlier work~\cite{hirvi2023effects} used both MCX and FEM to simulate FD measurements and the corresponding absorption sensitivity profiles. More recently, analytical expressions for TOF sensitivities in multilayered media have been derived under the DA and validated against MC simulations for time-resolved near-infrared spectroscopy (NIRS)~\cite{garcia2025analytical}. Ref.~\cite{cao2025nirfasterff} introduced a highly parallelized FEM solver for DA and validated it against MCX results. 

In comparison to MC, the DA offers faster but less accurate solutions that are valid only under more restrictive conditions~\cite{nissila2005diffuse}. The DA is derived by linearizing the radiance with respect to scattering direction and approximating the photon current with time-independency~\cite{arridge1999optical, heino2002estimation}. For the approximation to remain physically valid, scattering must dominate absorption, and light propagation should be weakly anisotropic, which is violated at least near sources where the angular distribution of photons becomes highly directional~\cite{schweiger1995finite}. The DA is nevertheless widely accepted as a sufficiently accurate model of light transport in highly scattering regimes, including most biological tissues encountered in medical imaging~\cite{arridge1999optical}. In brain imaging, however, different approaches have been suggested to overcome the challenge of modeling the low-scattering regions filled with cerebrospinal fluid (CSF)~\cite{nissila2005diffuse}, for example, combining radiosity theory to the DA~\cite{dehghaniVoids2000} or selecting alternative optical properties that match the resulting DA output with the MC reference~\cite{lewis2025revisiting}.  

A benefit of the MC method is its relative flexibility in modeling complex measurement systems, such as varying source and detector types. For example, MCX currently supports 18 different source types, which can incorporate various spatial and angular launch distributions. In this work, we focus on modeling light detection in MC. Previously, an imaging probe consisting of optical fibers, fiber bundles, and prism terminals embedded in a black silicone support was developed for the FD instrument at Aalto University~\cite{jonsson2018affective}. By modeling how light exiting the tissue near the detector optode positions and entering the prism can be transmitted to propagate along the detector fiber bundles, we implement a post-processing step ensuring that only photon packets that can realistically be detected contribute to the simulated measurable quantities and Jacobians. The selection procedure is essential for attempts to predict correct absolute levels for the detected optical power. It is also expected to affect the shape of the sensitivity profile, especially at short source--detector separations (SDSs), for example, by rejecting photons that arrive parallel to the head surface.

For improved modeling of the photon exit directions and positions as well as the curvature of the surface under the detector prism, the hybrid boundary mesh and voxelated inner content models are employed in the split-voxel MC mode (SVMC)~\cite{yan2020hybrid}. To the best of our knowledge, this is also the first time the SVMC approach is used for computing Jacobians in the ``replay'' mode.

The remainder of the paper is organized as follows. Section~\ref{sec:Methods} reviews the MC method for light transport and derives the MC Jacobian formulations for CW, FD, and TD measurements. It also presents the DA and adjoint-based approach used for comparison. Section~\ref{sec:implementation} describes the simulation settings, the finite element DA solver, and implementation details of the new Jacobian types. The operating principle and computational implementation of the prism-terminated optical fiber detectors in MC is explained. Section~\ref{sec:results} presents validation results and examines the effect and practical relevance of the angle- and position-limiting detector model in comparison to isotropic reception. Section~\ref{sec:discussion} discusses the findings and their implications, and Section~\ref{sec:conclusion} summarizes our key findings.

\section{Theory and Methods}
\label{sec:Methods}

In this work, the imaged region $\Omega \subseteq \mathbb{R}^3$ is modeled as a bounded Lipschitz domain with a connected complement~\cite{hannukainen2015edge}. For numerical computations, the region is discretised into voxels $\Omega_v$ forming a conforming mesh of $\Omega$. Light sources and detectors (optodes) are represented as two-dimensional, connected, mutually disjoint open subsets of the boundary $\partial \Omega$. Specifically, in this work the optodes are circular patches on the boundary.

The optical properties of the medium are described by four parameters: the absorption coefficient $\mu_a$ [mm$^{-1}$], the scattering coefficient $\mu_s$ [mm$^{-1}$], the scattering anisotropy $g$, and the refractive index $n$, all belonging to the set of essentially bounded positive functions
\begin{align*}
L^\infty_+(\Omega; \mathbb{R}):=  \left\{ v \in L^\infty (\Omega; \mathbb{R}) : \mathop{\operatorname{ess\,inf}} \ v  > 0\right\}. 
\end{align*}
In the voxel-based geometry, $\mu_a$ and $\mu_s$ are represented as piecewise constant functions:
\begin{align*}
\mu_a (\boldsymbol{x}) = \sum_v \mu_{a,v}\chi_{\Omega_v} (\boldsymbol{x}), \quad \mu_s (\boldsymbol{x}) = \sum_v  \mu_{s,v} \chi_{\Omega_v} (\boldsymbol{x}), \quad \boldsymbol{x}\in \overline{\Omega},
\end{align*}
where $\chi_{\Omega_v}$ is the characteristic function of the voxel $v$. In this work, the anisotropy $g$ and the refractive index $n$ are assumed constant in the entire domain (a relatively common assumption in head models). The latter implies a constant speed of light $c$ in the medium.

In the time domain, the RTE reads (cf.~\cite{arridge1999optical}):
\begin{align*}
&\left( \frac{1}{c}\frac{\partial }{\partial t}+\hat{\boldsymbol{s}}\boldsymbol{\cdot}\nabla_{\boldsymbol{x}}+\mu_a(\boldsymbol{x})+\mu_s(\boldsymbol{x}) \right) L(\boldsymbol{x}, \hat{\boldsymbol{s}}, t) \\
& \quad = \mu_s(\boldsymbol{x})\int_{\hat{\boldsymbol{s}}'\in \mathbb{S}^2} \Theta(\hat{\boldsymbol{s}}\boldsymbol{\cdot}\hat{\boldsymbol{s}}')L(\boldsymbol{x}, \hat{\boldsymbol{s}}', t)\,\mathrm{d}S(\hat{\boldsymbol{s}}'), \quad \boldsymbol{x}\in\Omega, \;
\hat{\boldsymbol{s}}\in\mathbb{S}^2, \;
t>0,
\end{align*}
where $L(\boldsymbol{x}, \hat{\boldsymbol{s}}, t)$ denotes the radiance at position $\boldsymbol{x}\in \Omega$, in direction $\hat{\boldsymbol{s}}\in \mathbb{S}^2$, at time $t$. In this work, the scattering phase function $\Theta$, i.e., the conditional probability density of scattering from direction $\hat{\boldsymbol{s}}'$ into direction $\hat{\boldsymbol{s}}$, is assumed to follow the Henyey--Greenstein model~\cite{henyey1941diffuse}:
\begin{align*}
\Theta(\hat{\boldsymbol{s}}\boldsymbol{\cdot}\hat{\boldsymbol{s}}') = \frac{1-g^2}{4\pi (1+g^2-2g \hat{\boldsymbol{s}} \boldsymbol{\cdot}\hat{\boldsymbol{s}}')^{3/2}}.
\end{align*}
This formulation of the RTE assumes no internal light sources. Boundary illumination and reflection are specified by suitable boundary conditions to the equation.

The CW RTE corresponds to the steady-state form of the time-domain equation, obtained by assuming a time-independent radiance. The frequency-domain equivalent of the TD RTE, in turn, is derived by assuming a complex-valued, time-harmonic radiance, $L(\boldsymbol{x}, \hat{\boldsymbol{s}}, t) = \hat{L}(\boldsymbol{x}, \hat{\boldsymbol{s}})e^{-\mathrm{i}\omega t}$, substituting this ansatz into the TD equation, and factoring out the time dependence (cf., e.g.,~\cite{arridge1999optical}). The negative exponential convention is adopted here so that increasing optical path lengths correspond to positive phase shifts in the detected signal.

\subsection{Microscopic Beer--Lambert law for Monte Carlo simulations}
\label{sec:pMC_CW}

MC simulations can estimate various quantities, including the probability of photon detection~\cite{sassaroli2012equivalence}. The detection probability $P$ is defined as the probability that a photon emitted within the source region is captured in the detection region. It can be interpreted as the measured CW intensity, up to a scaling factor; see further discussion in~\ref{sec:physic_quantities}. Since photon trajectories are sampled consistently with the RTE, $P$ corresponds to the detector measurement implied by the equation. In the following, the notation from~\cite{sassaroli2012equivalence} is adapted to express the detection probability formally.

Let $\mathcal{S}$ denote the space of all trajectories from the illumination source to the detector, each composed of a finite number of straight-line segments. A trajectory $\Gamma \in \mathcal{S}$ is represented by a sequence of ordered pairs $(\ell_i, \hat{\boldsymbol{s}}_i)$, $i = 0, 1, \dots, m(\Gamma)$, where $\hat{\boldsymbol{s}}_i \in \mathbb{S}^2$ is the direction of propagation after the $i$-th collision, $\ell_i \geqslant 0$ is the path length along $\hat{\boldsymbol{s}}_i$, and $m(\Gamma)$ is the number of collisions along the trajectory $\Gamma$. The detection probability can then be written as~\cite{sassaroli2012equivalence}:
\begin{align}
P(\mu_a, \mu_s) = \int_{\Gamma \in \mathcal{S}} \exp\left(-\int_{\Gamma}(\mu_a+\mu_s) \, \mathrm{d}s\right) \prod_{i=0}^{m(\Gamma)-1}\mu_s(\ell_i) \, \mathrm{d}\ell_i \Theta(\hat{\boldsymbol{s}}_{i+1} \boldsymbol{\cdot}\hat{\boldsymbol{s}}_{i}) \, \mathrm{d}S(\hat{\boldsymbol{s}}_{i+1}).
\label{eq:detection_probability}
\end{align}
Here, the notation $\mu_s(\ell_i)$ is used to denote the scattering coefficient at the location of the $i$-th collision along the trajectory (notation slightly abused for compactness). Dependencies on the scattering anisotropy $g$ and refractive index $n$ are not written explicitly, as they are assumed constant within the domain; derivatives with respect to these parameters are not considered in this work. Reflection and transmission coefficients arising from refractive index mismatches at the boundary are also omitted from Eq.~\eqref{eq:detection_probability} for conciseness, since they are independent of the absorption and scattering coefficients. Their inclusion is straightforward and does not affect the subsequent derivations or the MC estimates of the absorption and scattering Jacobians (cf.~\cite{amendola2024_part1} and Sec.~\ref{sec:mcx_implementation}).

In the mBLL method, the detection probability is written as~\cite{sassaroli2012equivalence}
\begin{align}
P(\mu_a, \mu_s) = \int_{\Gamma \in \mathcal{S}} W(\Gamma; \mu_a) \, \mathrm{d}\mathcal{P}(\Gamma; \mu_s),
\label{eq:detection_probability_mBLL}
\end{align}
where all absorption dependence is captured by the weight function
\begin{align}
W(\Gamma;\mu_a) = \exp \left(-\int_\Gamma \mu_a \, \mathrm{d}s\right)\, .
\label{eq:detection_weight_mBLL}
\end{align}
The remaining factors from the original expression for $P(\mu_a, \mu_s)$ in Eq.~\eqref{eq:detection_probability} are incorporated into the probability distribution $\mathcal{P}(\cdot; \mu_s)$ from which trajectories are sampled; see~\cite{sassaroli2012equivalence} for a more intuitive description. The sampled trajectories in the mBLL method are therefore independent of the absorption coefficient. In this formulation, the detection probability corresponds to the expectation of the random variable $W(\boldsymbol{\Gamma}; \mu_a)\chi_{\mathcal{S}}(\boldsymbol{\Gamma})$ when $\boldsymbol{\Gamma}$ is distributed according to $\mathcal{P}(\cdot; \mu_s)$.

In a voxel-based geometry, let us assume that $N$ trajectories are sampled (i.e., $N$ photon packets are simulated) i.i.d.\ from the distribution $\mathcal{P}(\cdot; \mu_s)$. The MC estimate of the expected value $P(\mu_a, \mu_s)$ is then the sample mean
\begin{align}
\hat{P}= \frac{1}{N}\sum_p w_p,
\label{eq:MC_probability}
\end{align}
where the summation runs over all photon packets $p$ that reach the detector, and the weight of each packet is defined as $w_p := W(\Gamma_p; \mu_a)$. In the voxelized domain, the weight can be expressed as
\[w_p = \mathrm{exp}\left(-\sum_{v} \mu_{a,v}\ell_{p,v}\right),\]
where the sum is over all voxels, and $\ell_{p,v}$ denotes the total path length of photon packet $p$ within voxel $v$~\cite{heiskala2007optical, leino2019valomc}. 

\subsubsection{Perturbation Monte Carlo for the microscopic Beer–Lambert law}

To derive MC-based expressions for the measurement sensitivities in different imaging modalities (FD and TD), it is convenient to first consider the derivative of the detection probability $P$ with respect to the voxel-wise absorption $\mu_{a,v}$ and scattering $\mu_{s,v}$ coefficients in voxel $v$. The derivative of $P$ with respect to the absorption coefficient can be written as
\begin{align}
\frac{\partial P(\mu_a, \mu_s)}{\partial \mu_{a,v}} = \int_{\Gamma \in\mathcal{S}} \frac{\partial W(\Gamma; \mu_a)}{\partial \mu_{a,v}} \, \mathrm{d}\mathcal{P}(\Gamma; \mu_{s}) = - \int_{\Gamma \in\mathcal{S}} L_v(\Gamma) W(\Gamma; \mu_{a}) \, \mathrm{d}\mathcal{P}(\Gamma; \mu_{s}) ,
\label{eq:probability_absorption}
\end{align}
where $L_v(\Gamma) := \vert \Gamma \cap \Omega_v\vert $ denotes the total path length of trajectory $\Gamma$ within voxel $v$. The corresponding MC estimate for the derivative is
\begin{align} \frac{\partial \hat{P}}{\partial \mu_{a,v}} = -\frac{1}{N}\sum_p \ell_{p,v} w_p,
\label{eq:MC_probability_absorption}
\end{align}
as also reported, e.g., in~\cite{heiskala2009accurate, heiskala2007optical}.

To derive the derivative of the detection probability with respect to voxel-wise scattering coefficients, the perturbation Monte Carlo (pMC) technique is utilized~\cite{yao2018direct, hayakawa2001perturbation, sassaroli1998monte, 9046017, amendola2024_part1}. Consider a perturbed scattering coefficient $\mu_s^\delta := \mu_s + \delta \mu_{s,v} \chi_{\Omega_v}$, where $\delta \mu_{s,v}$ is a constant perturbation applied to the scattering coefficient of voxel $v$. Using Eq.~\eqref{eq:detection_probability}, the perturbed detection probability can be expressed as
\begin{align}
P(\mu_a, \mu_s^\delta) &= \int_{\Gamma \in \mathcal{S}} \exp{\left(-\delta \mu_{s,v} L_v(\Gamma)\right)}\exp{\left(-\int_{\Gamma }(\mu_a+\mu_s) \, \mathrm{d}s\right)} \notag  \\
& \quad \times \prod_{i=0}^{m(\Gamma)-1} \left(\frac{\mu_s^\delta(\ell_i) }{\mu_s(\ell_i)}\right) \mu_{s}(\ell_i) \, \mathrm{d}\ell_i \Theta(\hat{\boldsymbol{s}}_{i+1} \boldsymbol{\cdot} \hat{\boldsymbol{s}}_i) \, \mathrm{d}S(\hat{\boldsymbol{s}}_{i+1}) \notag \\
&= \int_{\Gamma \in \mathcal{S}} \exp{\left(-\delta \mu_{s,v} L_v(\Gamma)\right)} \left(\frac{\mu_{s,v} + \delta \mu_{s,v}}{\mu_{s,v}}\right)^{m(\Gamma; \Omega_v)}W(\Gamma; \mu_a) \, \mathrm{d}\mathcal{P}(\Gamma; \mu_s),
\label{eq:inserted_mus_pert}
\end{align}
where $m(\Gamma; \Omega_v)$ denotes the number of collisions of trajectory $\Gamma$ within voxel $v$. The derivative of $P$ with respect to $\mu_{s,v}$ is then obtained by differentiating with respect to the perturbation and evaluating at zero:
\begin{align}
\frac{\partial P(\mu_a, \mu_s)}{\partial \mu_{s,v}} &= \left. \frac{\partial P(\mu_a, \mu_s^\delta)}{\partial (\delta \mu_{s,v})}\right|_{\delta\mu_{s,v}=0} \notag  \\
&= \int_{\Gamma \in \mathcal{S}} \left(\frac{m(\Gamma; \Omega_v)}{\mu_{s,v}}-L_v(\Gamma) \right) W(\Gamma; \mu_a) \, \mathrm{d}\mathcal{P}(\Gamma; \mu_s).
\label{eq:probability_scattering}
\end{align}
Accordingly, the MC estimate for the scattering derivative is
\begin{align}
\frac{\partial \hat{P}}{\partial \mu_{s,v}} = \frac{1}{N}\sum_p \left(\frac{m_{p,v}}{\mu_{s,v}}-\ell_{p,v}\right)w_p,
\label{eq:MC_probability_scattering}
\end{align}
where $m_{p,v}:=m(\Gamma_p; \Omega_v)$ is the number of collisions of photon packet $p$ in voxel $v$. This expression for the scattering derivative has the same form as previously applied for mBLL in FD~\cite{kangasniemi2024stochastic}, for mBLL in CW~\cite{yao2018direct}, and for albedo-weight in FD~\cite{yamamoto2016frequency}.

\subsubsection{Perturbation Monte Carlo for frequency-domain sensitivities}
\label{sec:pMC_FD}

As noted in~\cite{leino2019valomc}, for an intensity-modulated source with angular frequency $\omega$, the RTE can be interpreted as the CW RTE with a complex absorption coefficient $\mu_a - \mathrm{i}\frac{\omega}{c}$. Accordingly, the CW detection probability formula (Eq.~\ref{eq:detection_probability_mBLL}) can be extended to compute the complex FD intensity:  
\begin{align}
I_{\text{FD}}(\mu_a, \mu_s) &= \int_{\Gamma \in \mathcal{S}} W\left(\Gamma; \mu_a - \mathrm{i}\frac{\omega}{c}\right)\,\mathrm{d}\mathcal{P}(\Gamma; \mu_s)=\int_{\Gamma \in \mathcal{S}} W(\Gamma; \mu_a) e^{\mathrm{i}\omega |\Gamma|/c}\,\mathrm{d}\mathcal{P}(\Gamma; \mu_s) \notag  \\
&= \int_{\Gamma \in \mathcal{S}} W(\Gamma; \mu_a)\cos{\left(\omega \frac{|\Gamma|}{c}\right)}\,\mathrm{d}\mathcal{P}(\Gamma; \mu_s)+ \mathrm{i}\int_{\Gamma \in \mathcal{S}} W(\Gamma; \mu_a) \sin{\left(\omega \frac{|\Gamma|}{c}\right)}\,\mathrm{d}\mathcal{P}(\Gamma; \mu_s) \notag  \\
&= X(\mu_a, \mu_s) + \mathrm{i}Y(\mu_a, \mu_s).
\label{eq:FD_intensity}
\end{align}
The corresponding MC estimate is~\cite{heiskala2009accurate, heiskala2007optical}
\begin{align}
\hat{I}_{\text{FD}} = \frac{1}{N} \sum_{p} w_p  e^{\mathrm{i}\omega t_p} 
= \frac{1}{N} \sum_{p} w_p \cos(\omega t_p) 
+ \frac{\mathrm{i}}{N} \sum_{p} w_p \sin(\omega t_p)
= \hat{X} + \mathrm{i}\hat{Y},
\label{eq:MC_FD_intensity}
\end{align}
where $t_p := |\Gamma_p|/c$ is the TOF of the detected photon packet $p$. For a piecewise constant $c$, the conversion to TOF would occur already during the integration phase. 

Since the FD phase factor $e^{\mathrm{i}\omega |\Gamma|/c}$ is independent of the voxel-wise absorption and scattering coefficients, the respective derivatives for $I_{\text{FD}}$  follow from the same reasoning as in the CW case (Eqs.~\eqref{eq:probability_absorption} and~\eqref{eq:probability_scattering}):
\begin{align}
\frac{\partial I_{\text{FD}}(\mu_a, \mu_s)}{\partial \mu_{a,v}} &= - \int_{\Gamma \in \mathcal{S}} L_v(\Gamma) \, W(\Gamma; \mu_a) \, e^{\mathrm{i}\omega |\Gamma|/c} \, \mathrm{d}\mathcal{P}(\Gamma; \mu_s), 
\label{eq:x_and_y_mua_derivatives}
\\
\frac{\partial I_{\text{FD}}(\mu_a, \mu_s)}{\partial \mu_{s,v}} &= \int_{\Gamma \in \mathcal{S}} 
\Bigg( \frac{m(\Gamma; \Omega_v)}{\mu_{s,v}} - L_v(\Gamma) \Bigg) 
W(\Gamma; \mu_a) \, e^{\mathrm{i}\omega |\Gamma|/c} \, \mathrm{d}\mathcal{P}(\Gamma; \mu_s).
\label{eq:x_and_y_mus_derivatives}
\end{align}
The MC counterparts are given by
\begin{align}
\frac{\partial \hat{I}_{\text{FD}}}{\partial \mu_{a,v}} &= - \frac{1}{N} \sum_{p} \ell_{p,v} \, w_p \, e^{\mathrm{i}\omega t_p} 
= \frac{\partial \hat{X}}{\partial \mu_{a,v}} + \mathrm{i} \frac{\partial \hat{Y}}{\partial \mu_{a,v}}, 
\label{eq:MC_X_and_y_mua_derivatives}
\\
\frac{\partial \hat{I}_{\text{FD}}}{\partial \mu_{s,v}} &= \frac{1}{N} \sum_{p} \Big(\frac{m_{p,v}}{\mu_{s,v}} - \ell_{p,v}\Big) w_p \, e^{\mathrm{i}\omega t_p} 
= \frac{\partial \hat{X}}{\partial \mu_{s,v}} + \mathrm{i} \frac{\partial \hat{Y}}{\partial \mu_{s,v}}.
\label{eq:MC_X_and_y_mus_derivatives}
\end{align}

The FD quantities that are typically measured in practice, amplitude and phase~\cite{nissila2002instrumentation,nissila2005instrumentation}, can be computed from $X$ and $Y$ as~\cite{heiskala2009accurate, heiskala2007optical}
\[A = \sqrt{X^2 + Y^2}, \quad \varphi = \mathrm{atan2}(Y, X),\]
with MC estimates $\hat{A}$ and $\hat{\varphi}$ computed in the same way from $\hat{X}$ and $\hat{Y}$. Here, $\mathrm{atan2}(\cdot, \cdot)$ denotes the four-quadrant inverse tangent function. The derivatives of log-amplitude and phase with respect to any voxel-wise parameter $\mu_v \in \{\mu_{a,v}, \mu_{s,v}\}$ can then be expressed as~\cite{hirvi2023effects}
\begin{align}
\frac{\partial \ln A}{\partial \mu_v} = \frac{1}{A^2} \left( X \frac{\partial X}{\partial \mu_v} + Y \frac{\partial Y}{\partial \mu_v} \right), \quad \frac{\partial \varphi}{\partial \mu_v} = \frac{1}{A^2} \left( X \frac{\partial Y}{\partial \mu_v} - Y \frac{\partial X}{\partial \mu_v} \right),
\label{eq:log-amplitude_and_phase_derivatives}
\end{align}
and the MC estimates are obtained by substituting the theoretical quantities with their MC equivalents. 

The logarithm of the amplitude is often used in image reconstruction to balance changes induced by a given perturbation in amplitude data across a range of SDSs, i.e., to prevent short channels from disproportionately influencing the minimized error function and degrading localization accuracy~\cite{arridge1999optical, nissila2004, nissila2005instrumentation}. Other benefits of the log-scale are discussed in~\ref{sec:physic_quantities}.
 
\subsubsection{Perturbation Monte Carlo for time-domain sensitivities}
\label{sec:pMC_TD}

In the TD modality, ultrashort (picosecond) light pulses are emitted, and the distribution of photon times-of-flight is measured. As $P(\mu_a, \mu_s)$ in Eq.~\ref{eq:detection_probability} represents the probability of detecting a photon at any time, it can be extended to define the temporal point spread function (TPSF) as (cf.~\cite{sassaroli2012equivalence})
\begin{align*}
P(t; \mu_a, \mu_s) = \int_{\Gamma \in \mathcal{S}_t} W(\Gamma; \mu_a) \,\mathrm{d}\mathcal{P}(\Gamma; \mu_s),
\end{align*}
where the integral is taken over the subset of trajectories with TOF $t$, i.e.,
\begin{align*}
\mathcal{S}_t = \left\{\Gamma \in \mathcal{S} : |\Gamma| = ct\right\}.
\end{align*}
Integrating the TPSF over a time interval $\mathcal{T}$ yields the probability of detecting a photon within $\mathcal{T}$. In particular, integrating the TPSF over the entire time domain corresponds to the total detection probability,
\begin{align*} 
\int_0^\infty P(t; \mu_a, \mu_s)\mathrm{d}t=P(\mu_a, \mu_s).
\end{align*} 
Hence, the TPSF can be interpreted as an unnormalized probability density over photon arrival times.

In TD optical tomography, the time-integrated intensity and mean TOF, derived from the TPSF, can be used for image reconstruction of absorption and scattering changes~\cite{heiskala2009significance,mozumder2020evaluation}. The time-integrated intensity is defined analogously to the CW case as
\begin{align*}
I_{\text{TD}}(\mu_a, \mu_s) = P(\mu_a, \mu_s).
\end{align*}
Therefore, the voxel-wise derivatives of the intensity with respect to absorption and scattering coefficients, as well as the corresponding MC estimates, are given by the same expressions as in Eqs.~\eqref{eq:probability_absorption}--\eqref{eq:MC_probability_scattering} for the CW case.

The mean TOF represents the centroid of the TPSF and can be expressed as the ratio of its first and zeroth moments:
\begin{align}
\langle t\rangle(\mu_a, \mu_s)  = \frac{P^{(1)}(\mu_a, \mu_s)}{P(\mu_a,\mu_s)},
\label{eq:mean_TOF}
\end{align}
where
\begin{align}
P^{(1)}(\mu_a, \mu_s) &= \int_0^\infty t \, P(t; \mu_a, \mu_s) \, \mathrm{d}t = \int_{\Gamma \in \mathcal{S}} W(\Gamma; \mu_a) \frac{|\Gamma|}{c} \, \mathrm{d}\mathcal{P}(\Gamma; \mu_s)
\label{eq:moment}
\end{align}
is the first temporal moment of the TPSF. The corresponding MC estimate is
\begin{align}
\hat{P}^{(1)} = \frac{1}{N}\sum_p w_p t_p.
\label{eq:MC_moment}
\end{align}
Hence, an MC estimate for the mean TOF is (cf.~\cite{heiskala2009significance, steinbrink2001determining})
\begin{align*}
\hat{\langle {t}\rangle} = \frac{\hat{P}^{(1)}}{\hat{P}}= \frac{\sum_p w_p t_p}{\sum_p w_p}.
\end{align*}
Analogous to the CW case (Eqs.~\eqref{eq:probability_absorption} and \eqref{eq:probability_scattering}), the derivatives of the first temporal moment are 
\begin{align}
\frac{\partial P^{(1)}(\mu_a, \mu_s)}{\partial \mu_{a,v}} &= -\int_{\Gamma \in \mathcal{S}} L_v(\Gamma) W(\Gamma; \mu_a) \frac{|\Gamma|}{c} \, \mathrm{d}\mathcal{P}(\Gamma; \mu_s),
\label{eq:MC_moment_mua_derivative}
\\
\frac{\partial P^{(1)}(\mu_a, \mu_s)}{\partial \mu_{s,v}}  &= \int_{\Gamma \in \mathcal{S}}\left(\frac{m(\Gamma; \Omega_v)}{\mu_{s,v}}-L_v(\Gamma)\right)W(\Gamma; \mu_a)\frac{|\Gamma|}{c} \, \mathrm{d}\mathcal{P}(\Gamma; \mu_s),
\label{eq:MC_moment_mus_derivative}
\end{align}
and their MC counterparts are given by
\begin{align}
\frac{\partial \hat{P}^{(1)}}{\partial \mu_{a,v}}=-\frac{1}{N}\sum_p \ell_{p,v}w_pt_p, \quad \frac{\partial \hat{P}^{(1)}}{\partial \mu_{s,v}} = \frac{1}{N}\sum_p \left(\frac{m_{p,v}}{\mu_{s,v}}-\ell_{p,v}\right)w_p t_p.
\label{eq:MC_moment_derivative}
\end{align}
Differentiation of the quotient in Eq.~\eqref{eq:mean_TOF} yields the derivative of the mean TOF with respect to any voxel-wise parameter $\mu_v \in \left\{\mu_{a,v}, \mu_{s,v}\right\}$ as
\begin{align*}
\frac{\partial \langle t\rangle (\mu_a, \mu_s)}{\partial \mu_v} = \frac{1}{P(\mu_a, \mu_s)^2}\left(\frac{\partial P^{(1)}(\mu_a, \mu_s)}{\partial \mu_v}P(\mu_a, \mu_s) - P^{(1)}(\mu_a, \mu_s)\frac{\partial P(\mu_a, \mu_s)}{\partial \mu_v}\right),
\end{align*}
and the corresponding MC estimates are obtained by substituting the theoretical quantities with their MC counterparts (Eqs.~\eqref{eq:MC_probability}, \eqref{eq:MC_probability_absorption}, \eqref{eq:MC_probability_scattering}, \eqref{eq:MC_moment} and \eqref{eq:MC_moment_derivative}), resulting in
\begin{align}
\frac{\partial \hat{\langle {t}\rangle} }{\partial \mu_{a,v}} &= -\frac{\sum_p w_p \ell_{p,v} t_p}{\sum_p w_p} +\frac{\sum_p w_p t_p}{\sum_p w_p} \cdot \frac{\sum_p w_p \ell_{p,v} }{\sum_p w_p} \notag \\
&= -\langle \ell_v t \rangle +\hat{\langle {t}\rangle} \langle \ell_v\rangle 
\label{eq:meantof_mua_derivative}
\end{align}
and
\begin{align}
\frac{\partial\hat{\langle {t}\rangle} }{\partial \mu_{s,v}} &=\frac{1}{\mu_{s,v}}\cdot  \frac{\sum_p w_pm_{p,v}t_p}{\sum_p w_p}-\frac{\sum_p w_p \ell_{p,v}t_p}{\sum_p w_p} \notag \\
& \quad -\frac{\sum_p w_p t_p}{\sum_p w_p} \cdot \left(\frac{1}{\mu_{s,v}} \cdot \frac{\sum_p w_p m_{p,v}}{\sum_p w_p}-\frac{\sum_p w_p \ell_{p,v}}{\sum_p w_p}\right) \notag \\
&= \frac{\langle m_v t\rangle }{\mu_{s,v}} - \langle \ell_{v}t\rangle -\hat{\langle {t}\rangle} \left(\frac{\langle m_v\rangle }{\mu_{s,v}}-\langle \ell_v\rangle \right), 
\label{eq:meantof_mus_derivative}
\end{align} 
where the angle brackets denote weighted averages.

\subsubsection{Comparing frequency- and time-domain quantities}
\label{sec:FD_vs_TD}

The relationship between the FD and TD quantities is considered for an intensity-modulated source with a modest angular frequency~$\omega$. If the TOF $t_p$ of a photon packet $p$ satisfies $|\omega (t_p - \langle \hat{t}\rangle)| \ll 1$, the complex exponential can be linearized as
\begin{align*}
e^{\mathrm{i}\omega t_p} \simeq e^{\mathrm{i}\omega \hat{\langle {t}\rangle}}\left(1 + \mathrm{i}\omega (t_p - \hat{\langle {t}\rangle})\right).
\end{align*}
If photon packets with the largest weights satisfy this condition, substituting the linearization into the MC estimate of the FD intensity (Eq.~\ref{eq:MC_FD_intensity}) yields 
\begin{align*}
\hat{I}_{\text{FD}} = \frac{1}{N} \sum_p w_p e^{\mathrm{i}\omega t_p} \simeq \frac{e^{\mathrm{i}\omega \hat{\langle {t}\rangle}} }{N}\sum_p w_p = \hat{I}_{\text{TD}} e^{\mathrm{i}\omega \hat{\langle {t}\rangle}}.
\end{align*}
The approximation is therefore valid when $| \omega (t_p - \langle \hat{t}\rangle )| \ll 1$ for photon packets with significant weight, which is satisfied when the TOF distribution is narrow, TOFs are short, or the modulation frequency (typically $10^8$--$10^9$\,Hz) is sufficiently low. Hence, the amplitude and phase can be approximated as
\begin{align*}
\hat{A} = |\hat{I}_{\text{FD}}| \simeq \hat{I}_{\text{TD}}, \quad \hat{\varphi} = \mathrm{atan2}(\Im{(\hat{I}_{\text{FD}})}, \Re{(\hat{I}_{\text{FD}})}) \simeq \omega  \hat{\langle {t}\rangle} \bmod (-\pi, \pi],
\end{align*}
where $\bmod\, (-\pi, \pi]$ denotes reduction modulo $2\pi$ into the interval $(-\pi, \pi]$. Thus, the FD amplitude is approximately equal to the TD integrated intensity, and the phase corresponds to the mean TOF scaled by the modulation angular frequency~$\omega$. Consequently, the corresponding sensitivity profiles of these quantities are expected to exhibit similar characteristics, at least up to frequencies less than 200--400\,MHz~\cite{arridge1992theoretical}.

\subsubsection{Perturbation Monte Carlo generalized to multi-label elements}
\label{sec:pMC_generalized}

All presented MC Jacobian formulas can be directly generalized for alternatively discretized domains by simply replacing the voxel $v$ with any (reasonably small) element, such as, a split voxel or tetrahedron. For simplicity, we will keep denoting these more general elements with $v$. 

Extending to elements with multiple labels, i.e., piecewise constant optical coefficients requires a brief revision of the presented formulas. For the purposes of this manuscript, we will only consider elements split into two parts, but the following generalizes to denser divisions. However, the per-element coefficient changes must be assumed constant. 

First, we need to introduce the pMC formula for the absorption Jacobians. Consider a perturbed absorption coefficient $\mu_a^\delta := \mu_a + \delta \mu_{a,v} \chi_{\Omega_v}$, where $\delta \mu_{a,v}$ is a constant perturbation applied to the piecewise constant absorption coefficient of element $v$. Let us denote the CW, FD and TD intensities and the first moment of the TPSF generally with $P_F \in \{P, I_{\text{TD}}, I_{\text{FD}}, {P}^{(1)}\}$. Inserting $\mu_a^\delta$ into Eq.~\eqref{eq:detection_probability_mBLL}, \eqref{eq:FD_intensity} or \eqref{eq:moment} yields
\begin{align*}
P_F(\mu_a^\delta, \mu_s) &= \int_{\Gamma \in \mathcal{S}} \exp{\left(-\delta \mu_{a,v} L_v(\Gamma)\right)}W(\Gamma; \mu_a)F(\Gamma) \, \mathrm{d}\mathcal{P}(\Gamma; \mu_s)\, , \quad 
F(\Gamma) =
\begin{cases}
\begin{alignedat}{2}
1, & &\quad &\; \text{for } P,\ I_{\text{TD}},  \vspace{1.1em}\\
e^{\mathrm{i}\omega |\Gamma|/c}, & &\quad &\;\text{for } I_{\text{FD}},  \vspace{1.1em}\\
|\Gamma|/c, & &\quad &\; \text{for } {P}^{(1)}, \vspace{1.1em}
\end{alignedat}
\end{cases}
\end{align*}
where $L_v(\Gamma)$ again denotes the total path length of trajectory $\Gamma$ within element $v$. Differentiating with respect to the perturbation and evaluating at zero results in the same derivatives as in Eqs.~\eqref{eq:probability_absorption}, \eqref{eq:x_and_y_mua_derivatives} and \eqref{eq:MC_moment_mua_derivative}, thus the MC estimate is also the same as in Eqs.~\eqref{eq:MC_probability_absorption}, \eqref{eq:MC_X_and_y_mua_derivatives} and \eqref{eq:MC_moment_derivative}. The essential difference is that we did not require the per-element absorption coefficient to be constant, which justified the original direct differentiation.

The scattering Jacobians require a slightly more careful examination. Let $\Omega_{v_1}$ and $\Omega_{v_2}$ denote the two sub-elements with scattering coefficients $\mu_{s,v_1}$ and $\mu_{s,v_2}$, respectively. We again consider a constant change in the scattering coefficient $\mu_s^\delta := \mu_s + \delta \mu_{s,v} \chi_{\Omega_v}$ over the whole element. Similarly as above, inserting this into Eq.~\eqref{eq:detection_probability_mBLL}, \eqref{eq:FD_intensity} or \eqref{eq:moment} yields 
\begin{align}
P_F(\mu_a, \mu_s^\delta) &= \int_{\Gamma \in \mathcal{S}} \exp{\left(-\delta \mu_{s,v} L_v(\Gamma)\right)} \left(\frac{\mu_{s,v_1} + \delta \mu_{s,v}}{\mu_{s,v_1}}\right)^{m(\Gamma; \Omega_{v_1})} \left(\frac{\mu_{s,v_2} + \delta \mu_{s,v}}{\mu_{s,v_2}}\right)^{m(\Gamma; \Omega_{v_2})} \notag \\
& \quad \quad \times W(\Gamma; \mu_a)F(\Gamma) \, \mathrm{d}\mathcal{P}(\Gamma; \mu_s),
\label{eq:general_probability_absorption}
\end{align}
where $m(\Gamma; \Omega_{v_1})$ and $m(\Gamma; \Omega_{v_2})$ are the total numbers of collisions along trajectory $\Gamma$ within the two sub-elements.

Applying once more the strategy of differentiating with respect to $\delta \mu_{s,v}$ and evaluating at $\delta\mu_{s,v} = 0$, similarly to the steps in Eqs.~\eqref{eq:inserted_mus_pert}--\eqref{eq:probability_scattering}, we end up at
\begin{align}
\frac{\partial P_F(\mu_a, \mu_s)}{\partial \mu_{s,v}} &= \left. \frac{\partial P_F(\mu_a, \mu_s^\delta)}{\partial (\delta \mu_{s,v})}\right|_{\delta\mu_{s,v}=0} \notag  \\
&= \int_{\Gamma \in \mathcal{S}} \left(\frac{m(\Gamma; \Omega_{v_1})}{\mu_{s,v_1}} + \frac{m(\Gamma; \Omega_{v_2})}{\mu_{s,v_2}} - L_v(\Gamma) \right) W(\Gamma; \mu_a)\, F(\Gamma) \, \mathrm{d}\mathcal{P}(\Gamma; \mu_s).
\label{eq:general_probability_scattering}
\end{align}
Consequently, the derivatives resemble Eqs.~\eqref{eq:probability_scattering}, \eqref{eq:x_and_y_mus_derivatives} and \eqref{eq:MC_moment_mus_derivative}, the only difference being that the scattering counts in each sub-element have to be scaled with the respective scattering coefficient. The same difference applies to computing the MC correspondents.  

\subsubsection{Further discussion regarding physical quantities}
\label{sec:physic_quantities}

Alternative scaling schemes can be used to interpret detected photon packet weights as physically meaningful quantities. The detection probability~\eqref{eq:MC_probability} can be interpreted as the detected energy or power for a unit source. Consequently, when attenuation is defined as the base-10 logarithmic ratio of the input to the detected energy or power, it can be estimated from the detection probability as $-\log_{10} {\hat{P}}$. This formulation also applies to intensity under the (Beer-)Lambert law~\cite{cope1991development}, and more generally depending on instrumentation and calibration. 

Dividing the detection probability or the complex intensity in Eq.~\eqref{eq:MC_FD_intensity} by the detector area is used to estimate the diffuse reflectance~\cite{yao2018direct} or exitance for a unit-strength source~\cite{leino2019valomc, kangasniemi2024stochastic}, where detection direction can also be incorporated~\cite{yamamoto2016frequency}. Yet, the true ``detector area'' depends on the instrument and may differ from the reception area on the surface (see Sec.~\ref{sec:detector_implementation}). Generally, realistic source and detector models are required to estimate actual signal strengths. On the other hand, the raw detector outputs are often not inherently expressed in (the wanted) physical units. One solution is to calibrate the output by matching measured absolute data to model predictions~\cite{nissila2005instrumentation}. Simulated values for a different model can then be shifted to the calibrated baseline~\cite{maria2022imaging}. 

In the end, constant scaling or normalization of the complex intensities in Eq.~\eqref{eq:MC_FD_intensity}, including normalization by the number of simulated photon packets, effectively cancels out in log-amplitude and phase difference and derivatives, and can thus be omitted as in~\cite{hirvi2023effects}. As a result, the source power or any positive constant multiplicative factor does not affect reconstruction from log-amplitude and phase difference data and can therefore be omitted. This also holds for TD log-intensity and mean TOF, thus all measurements that are observed in more detail in this work. 

\subsection{Diffusion approximation in the frequency domain}

The DA describes light propagation in terms of the photon density (or fluence) rather than the radiance. In the FD case, the photon density is obtained from the radiance $\hat{L}$ by integrating over all directions as~\cite{arridge1999optical, heino2002estimation}
\begin{align*}
\Phi(\boldsymbol{x})=\int_{\hat{\boldsymbol{s}} \in \mathbb{S}^2} \hat{L}(\boldsymbol{x}, \hat{\boldsymbol{s}}) \, \mathrm{d}S(\hat{\boldsymbol{s}}).
\end{align*}
Under the DA to the FD RTE, the photon density $\Phi_k \in H^1(\Omega; \mathbb{C})$ resulting from illumination at the $k$-th source is the unique~\cite{grisvard2011elliptic} solution to the following elliptic Robin boundary value problem (cf.~\cite{arridge1999optical, heino2002estimation, hirvi2023effects, schweiger1995finite, hannukainen2015edge}):

\begin{align}
\begin{cases}
\begin{alignedat}{2}
- \nabla \boldsymbol{\cdot} \left(\kappa \nabla \Phi_k\right) 
  + \Big(\mu_a - \mathrm{i}\frac{\omega}{c}\Big) \Phi_k
  &\;= 0, & &\quad \text{in } \Omega,  \vspace{1.0em}\\
\frac{1}{4}(1-\rho)\Phi_k + \frac{1}{2}(1+\rho)\, \boldsymbol{\nu}\boldsymbol{\cdot}\kappa \nabla \Phi_k
  &\;= Q_k,& & \quad \text{on } \partial \Omega.
\end{alignedat}
\end{cases}
\label{eq:diffusion_approximation}
\end{align}
Here, $\boldsymbol{\nu}\text{: }\partial \Omega \rightarrow \mathbb{S}^2$ is the outward unit normal of $\partial \Omega$ and $\rho$ is the effective reflection coefficient~\cite{Haskell:94} accounting for the Fresnel reflection on the boundary. The diffusion coefficient $\kappa \in L^\infty_+(\Omega; \mathbb{R})$ is defined in terms of the absorption, scattering, and anisotropy coefficients as~\cite{arridge1999optical, nissila2005diffuse,heino2002estimation}
\begin{align*}
\kappa = \frac{1}{3[\mu_a + (1-g)\mu_s]}.
\end{align*}
In FD, the diffusion coefficient is in some cases modeled as complex-valued~\cite{heino2002estimation}; however, no significant differences were observed in the present work when the two formulations were compared. The boundary source is included via the boundary condition by modeling the input flux $Q_k\text{: }\partial \Omega \rightarrow \mathbb{R}$ as a characteristic function of the $k$-th source divided by its area.

The weak formulation corresponding to \eqref{eq:diffusion_approximation} is obtained by multiplying the diffusion equation by a test function $\psi$, integrating over $\Omega$ by parts, and incorporating the Robin boundary condition. The weak problem reads: find a photon density $\Phi_k \in H^1(\Omega; \mathbb{C})$ such that the variational equality \begin{align} \int_\Omega \left(\kappa \nabla \Phi_k \boldsymbol{\cdot} \nabla \overline{\psi} + \left(\mu_a -\mathrm{i}\frac{\omega}{c}\right)\Phi_k \overline{\psi}\right) \, \mathrm{d}x + \frac{1-\rho}{2(1+\rho)} \int_{\partial \Omega} \Phi_k \overline{\psi} \, \mathrm{d}S = \frac{2}{1+\rho} \int_{\partial \Omega} Q_k \overline{\psi} \, \mathrm{d}S \label{eq:weak_problem} \end{align} holds for all test functions $\psi \in H^1(\Omega; \mathbb{C})$~\cite{heino2002estimation, hirvi2023effects, hannukainen2015edge}.

The detected flux at the $j$-th detector can be evaluated as~\cite{heino2002estimation, hirvi2023effects, hannukainen2015edge}
\[\mathcal{M}_{jk} = \int_{\partial \Omega} P_j(1-\rho) \left(\frac{1}{4} \Phi_k -\frac{1}{2}\boldsymbol{\nu} \boldsymbol{\cdot}\kappa \nabla \Phi_k\right) \, \mathrm{d}S = \frac{1-\rho}{2(1+\rho) } \int_{\partial \Omega} P_j \Phi_k \, \mathrm{d}S,\]
where $P_j$ is the characteristic function of the detector.

\subsubsection{Adjoint method}

The Fr\'{e}chet derivatives of the measurement $\mathcal{M}_{jk}$ with respect to the absorption and scattering coefficients can be computed using the adjoint method~\cite{arridge1999optical, hannukainen2015edge, arridge1995photon}. To derive formulas for the Fr\'{e}chet derivatives, the following adjoint problem is introduced: find $\Phi_j^* \in H^1(\Omega; \mathbb{C})$ such that
\begin{align}
\begin{cases}
\begin{alignedat}{2}
- \nabla \boldsymbol{\cdot} \left(\kappa \nabla \Phi_j^*\right) 
  + \Big(\mu_a - \mathrm{i}\frac{\omega}{c}\Big) \Phi_j^*
  &\;= 0, & &\quad \text{in } \Omega,  \vspace{1.0em}\\
\frac{1}{4}(1-\rho)\Phi_j^*+ \frac{1}{2}(1+\rho)\, \boldsymbol{\nu}\boldsymbol{\cdot}\kappa \nabla \Phi_j^*
  &\;= P_j,& & \quad \text{on } \partial \Omega.
\end{alignedat}
\end{cases}
\label{eq:adjoint_problem}
\end{align}
Here, the boundary input flux $Q_k$ is replaced by the detector function $P_j$, and the weak formulation is derived in the same manner as for the forward problem \eqref{eq:diffusion_approximation}. 

The Fr\'{e}chet derivatives can now be expressed explicitly via the forward and adjoint solutions (cf.~\cite{hannukainen2015edge}).

\begin{theorem}
Let $\Phi_k$ and $\Phi^*_j$ be the unique weak solutions of \eqref{eq:diffusion_approximation} and \eqref{eq:adjoint_problem}, respectively. The Fr\'{e}chet derivative of the measurement map $(\mu_a, \mu_s) \mapsto \mathcal{M}_{jk}$ at $(\mu_a, \mu_s) \in [L_+^\infty(\Omega; \mathbb{R})]^2$ in the direction $(\zeta, \vartheta) \in [L^\infty(\Omega; \mathbb{R})]^2$ can be evaluated as
\begin{align*}
D\mathcal{M}_{jk}(\mu_a, \mu_s)[(\zeta, \vartheta)] &=  \frac{1}{4} (1-\rho) \int_{\Omega } \zeta \left(3\kappa^2 \nabla \Phi_k \boldsymbol{\cdot} \nabla \Phi_j^*-\Phi_k \Phi_j^*\right) \,\mathrm{d}x \\
& \quad + \frac{3}{4} (1-\rho) \int_\Omega \vartheta \kappa^2(1-g) \nabla \Phi_k \boldsymbol{\cdot} \nabla \Phi_j^* \, \mathrm{d}x.
\end{align*}
\label{thm:Fréchet_derivatives}
\end{theorem}
The proof of Theorem \ref{thm:Fréchet_derivatives} is provided in Appendix~\ref{sec:proof}.

In particular, the derivatives of the measurement with respect to voxel-wise absorption and scattering coefficients are obtained as
\begin{align}
\frac{\partial \mathcal{M}_{jk}}{\partial \mu_{a,v}} &= \frac{1}{4}(1-\rho)\int_{\Omega_v}  \left(3\kappa^2 \nabla \Phi_k \boldsymbol{\cdot} \nabla \Phi_j^*-\Phi_k \Phi_j^*\right) \,\mathrm{d}x, \label{eq:DA_absorption_derivative} \\
\frac{\partial \mathcal{M}_{jk}}{\partial \mu_{s,v}} &= \frac{3}{4}(1-\rho)\int_{\Omega_v} \kappa^2 (1-g)\nabla \Phi_k \boldsymbol{\cdot}\nabla \Phi_j^* \, \mathrm{d}x. \label{eq:DA_scattering_derivative}
\end{align}

\section{Implementation and modeled instrumentation}
\label{sec:implementation}

\subsection{Models}

Anatomical models were selected for one full-term neonate at a combined gestational and chronological age of 41.3 weeks, and one late preterm neonate at 35.1 weeks, from the database of head models measured under the Developing Human Connectome Project and processed in collaboration between University College London and the Centre for the Developing Brain at King's College London~\cite{CollinsJones2021}. The orientation and segmentation of the models were modified as described in Sections 3.1--3.2 of~\cite{hirvi2023effects} and Section 5.1 of~\cite{hakula2026bilinear}. Following the modifications, each three-dimensional pixel, or voxel $v$ in the head model was assigned to a single tissue type: combined scalp and skull (S\&S), cerebrospinal fluid (CSF), grey matter (GM), or white matter (WM). The optical parameters selected from the literature for each tissue type at the wavelength of 798~nm are listed in Table~\ref{tab:optical_properties}. The modulation frequency was set to 100\,MHz according to the FD instrument at Aalto University. 

\begin{table}[b!] 
\centering
\caption{Optical properties of tissue types in the neonatal head model at 798\, nm~\cite{fukui2003,jonsson2018affective,hirvi2023effects}. Cerebrospinal fluid = CSF.}
\label{tab:optical_properties}
\begin{tabular}{l c c c c}
\hline
Tissue Type & $\mu_a$ [mm$^{-1}$] & $\mu_s$ [mm$^{-1}$] & $g$ & $n$ \\
\hline
Scalp \& Skull & 0.015 & 16.0 & 0.9 & 1.4 \\
CSF (Subarachnoid) & 0.004 & 1.6 & 0.9 & 1.4 \\
CSF (Sulci, Ventricles) & 0.002 & 0.4 & 0.9 & 1.4 \\
Grey Matter & 0.048 & 5.0  & 0.9 & 1.4 \\
White Matter & 0.037 & 10.0 & 0.9 & 1.4 \\
\hline
\end{tabular}
\end{table}

To simulate Jacobians in MCX for comparison to results according to the DA, cubic voxels with physical dimensions $1 \times 1 \times 1$~mm$^3$ were used and all CSF voxels were modeled as semi-diffusive (values for subarachnoid in Table~\ref{tab:optical_properties}). The optodes were roughly circular patches on the boundary with radii $1.5$~mm (sources) and $2.0$~mm (detectors). In MCX, the detection area is the intersection of a sphere with the selected radius and the head surface. Isotropic reception was used in the MC versus DA comparisons. Back-reflections of input light at source locations were neglected, but internal reflections were modeled on exterior boundary. No reflection or refraction events occur within the medium, since the refractive index is assumed constant across all tissue types.

\subsection{Monte Carlo eXtreme}
\label{sec:mcx_implementation}

The v2025.10 release of MCXLAB (pre-compiled) was used for all MC computations, except for the SVMC simulations described in Sec.~\ref{sec:svmc}. All MC simulations were performed on Triton, which is a high-performance computing cluster provided by the School of Science at Aalto University. Each job was randomly assigned one 16--80~GB graphical memory (VRAM) NVIDIA V100/A100/H100 GPU and 1--4 CPU on the 64-bit Linux operating system.  

The Jacobians were computed as the algebraic combinations of various photon metrics obtained separately in the replay mode. All required per-voxel metrics were computed from the same set of sampled photon trajectories in the replay mode, requiring only a single time-consuming forward simulation~\cite{yao2018direct}. The FD Jacobians from Eqs.~\eqref{eq:MC_X_and_y_mua_derivatives}--\eqref{eq:MC_X_and_y_mus_derivatives} were obtained by combining the (complex) per-voxel total weighted path lengths ($\partial \hat{X}/\partial \mu_{a,v}$ and $\partial \hat{Y}/\partial \mu_{a,v}$, as the real and imaginary parts of MCX output type ``\texttt{rf}'', respectively) established in~\cite{hirvi2023effects}, with the per-voxel total weighted scattering counts (numbers of collisions) implemented in this work (output \#1, corresponding to MCX output type ``\texttt{rfmus}''). The TD Jacobians from Eqs.~\eqref{eq:MC_probability_absorption}, \eqref{eq:MC_probability_scattering}, \eqref{eq:meantof_mua_derivative} and \eqref{eq:meantof_mus_derivative} required the pure~\cite{yao2018direct} and the TOF-weighted per-voxel average path lengths (output \#2, $\langle t \ell_{v}\rangle$, MCX output type ``\texttt{wltof}'') and scattering counts (output \#3, $\langle t m_{v}\rangle$, MCX output type ``\texttt{wptof}'') implemented in this work. Mean TOF ($\hat{\langle t \rangle}$) and all absolute data types were computed from the tissue-wise partial paths recorded for each detected photon packet. As a tutorial, the following revises the used Jacobian formulas in combination with the output names for each term in MCX, highlighting the new replay output quantities \#1--\#3:  
\begin{equation}
\begin{cases}
\dfrac{\partial \hat{X}}{\partial \mu_{s,v}} &= \dfrac{1}{N \mu_{s,v}}\, \underbrace{ \displaystyle \sum\limits_{p}^{} \left[ w_{p}\, m_{p,v}\, \cos{(\omega t_{p})} \right]}_{\text{\textbf{\#1} ``\texttt{rfmus}''}}\, \underbrace{ \underbrace{ - \sum\limits_{p}^{} \left[ w_{p}\, \ell_{p,v}\, \cos{(\omega  t_{p})} \right]}_{\text{``\texttt{rf}'' \cite{hirvi2023effects}}}\, \dfrac{1}{N}}_{{\partial \hat{X}}/{\partial \mu_{a,v}}}\, , \vspace{1.0em}\\
\dfrac{\partial \hat{Y}}{\partial \mu_{s,v}} &= \dfrac{1}{N \mu_{s,v}}\, \underbrace{ \displaystyle \sum\limits_{p}^{} \left[ w_{p}\, m_{p,v}\, \sin{(\omega  t_{p})} \right]}_{\text{\textbf{\#1} ``\texttt{rfmus}''}}\, \underbrace{ \underbrace{ - \sum\limits_{p}^{} \left[ w_{p}\, \ell_{p,v}\, \sin{(\omega  t_{p})} \right]}_{\text{``\texttt{rf}'' \cite{hirvi2023effects}}}\, \dfrac{1}{N}}_{{\partial \hat{Y}}/{\partial \mu_{a,v}}}\, , \vspace{1.0em}\\ 
\dfrac{\partial \hat{\langle t \rangle}}{\partial \mu_{a,v}} &= - \underbrace{\langle t\, \ell_{v} \rangle}_{\text{\textbf{\#2} ``\texttt{wltof}''}}\, +\, \hat{\langle t \rangle}\, \underbrace{ \underbrace{ \langle \ell_{v}\rangle}_{\text{``\texttt{jacobian}'' \cite{yao2018direct}}}}_{- {\partial \ln \hat{I}_{\text{TD}}}/{\partial \mu_{a,v}}}\, , \vspace{1.0em}\\
\dfrac{\partial \hat{\langle t \rangle}}{\partial \mu_{s,v}} &= \dfrac{1}{\mu_{s,v}}\, \underbrace{\langle t\, m_{v} \rangle}_{\text{\textbf{\#3} ``\texttt{wptof}''}}\, -\, \underbrace{\langle t\, \ell_{v} \rangle}_{\text{\textbf{\#2} ``\texttt{wltof}''}}\, -\, \hat{\langle t \rangle}\, \underbrace{ \left( \dfrac{1}{\mu_{s,v}}\, \underbrace{\langle m_{v} \rangle}_{\text{``\texttt{nscat}'' \cite{yao2018direct}}}\, -\, \underbrace{\langle \ell_{v} \rangle}_{\text{``\texttt{jacobian}'' \cite{yao2018direct}}} \right)}_{{\partial \ln \hat{I}_{\text{TD}}}/{\partial \mu_{s,v}}}\, .
\end{cases}
\label{eq:MCXimplementation}
\end{equation}
A demonstration script illustrating the computation of the Jacobians is also provided in~\cite{replay_demo}. It should be noted that the expressions of the scattering Jacobians shown in Eq.~\ref{eq:MCXimplementation} always contain the respective absorption Jacobian on the right-hand-side.  

Each individual term in Eq.~\ref{eq:MCXimplementation} is essentially a sum over detected photon packets. Since the Triton cluster provides multiple GPUs for embarrassingly parallel computations, simulations were performed in batches with a reduced number of photon packets, and the independent contributions were accumulated in post-processing. This reduction was also observed to minimize the risk of random failure due to overheating or other reasons. Normalization by the total detected weight for quantities denoted by angle brackets was applied after aggregation, and the MCX default normalization was disabled. The appropriately aggregated FD components were converted to log-amplitude and phase Jacobians using Eq.~\eqref{eq:log-amplitude_and_phase_derivatives}. Simulations were performed using $(5 \times 10^{6})$--$10^{9}$ photon packets per batch, and iterated for a total of $(2 \times 10^{10})$--$10^{13}$ packets, as detailed later for each case. Job-specific random number generator seeds were set to avoid simulating the exact same trajectories. The maximum TOF was limited to one modulation period determined by the modulation frequency of the FD source. 

MC estimators of the RTE-based measurements converge to their corresponding quantities of interest as the number of sampled trajectories increases. The limited sample size introduces stochastic noise, manifesting as pixelation in the Jacobians and affecting, for example, the accuracy of surface-related phenomena: in our simulations, at the head--air interface, a photon packet either exited and was terminated, or reflected back into the medium without weight loss upon reflection. The back-reflection probability was determined by the reflection coefficient. As the number of sampled trajectories increased, the reflection probability was expressed via the number of realizations of the same trajectory. The reflection coefficient depends only on the refractive indices and the angle of incidence governed by the Henyey--Greenstein phase function through $g$. This background supports omitting the reflection coefficient from Eq.~\eqref{eq:detection_probability} onward. 

\subsection{Diffusion approximation solver}

Reference Jacobians based on the DA were computed using an in-house finite element (FE) solver. The solver was originally developed in~\cite{hirvi2023effects} for computing FD absorption Jacobians using scikit-fem~\cite{skfem2020}, and was extended in the present work to also support FD scattering Jacobians. The solver is used here to numerically solve the weak formulation of the FD DA~\eqref{eq:weak_problem} and to compute the corresponding Jacobians as described in Theorem~\ref{thm:Fréchet_derivatives}. The implementation is independent of MCX.

The FE formulation employs a standard first-order Lagrange ($\mathbb{P}_1$) basis defined on a three-dimensional tetrahedral mesh. The FE meshes were constructed directly from the voxel models used in the MC simulations by subdividing each voxel into six tetrahedra and subsequently refining the mesh, particularly in regions near the optodes. This procedure preserves the exact geometry of the original voxel model while enabling a conforming tetrahedral discretization suitable for FEM.

Initially, each tetrahedral element inherited its optical properties from its parent voxel, resulting in piecewise constant parameter fields. To obtain visually smoother Jacobians that are more comparable with those produced by MC simulations, nodal values for each optical parameter were defined as the arithmetic mean of the values in the neighboring tetrahedra, and the parameters were redefined as linear interpolants of these nodal values. While sharp changes in the Jacobians at tissue boundaries are expected in FEM with discontinuous optical parameters, the MC Jacobians tend to appear smoother due to stochastic averaging. The interpolation step therefore serves primarily to facilitate visual comparison between the two approaches.

The adjoint method was applied to compute element-wise absorption and scattering derivatives of the complex intensity by evaluating the integrals in \eqref{eq:DA_absorption_derivative}–\eqref{eq:DA_scattering_derivative} over individual tetrahedral elements of the FE mesh rather than directly over the full voxel domains. The resulting element-wise derivatives were then aggregated by summing the contributions of all tetrahedra associated with the same parent voxel. This summation effectively realizes the integrals in \eqref{eq:DA_absorption_derivative}–\eqref{eq:DA_scattering_derivative} over the voxel $\Omega_v$. Finally, separating the real and imaginary parts of the resulting complex Jacobians yields the log-amplitude and phase Jacobians using formulas~\eqref{eq:log-amplitude_and_phase_derivatives}.

\subsection{Detector models in Monte Carlo} 
\label{sec:detector_implementation}

In MC, the recorded exit positions and directions of detected photon packets can be utilized to weight or limit the packets that contribute to the final output. This sub-section investigates accurate modeling of a prism-terminated fiber detector using MC simulations. This type of detector has been used in several fNIRS systems, including the FD system at Aalto University, developed at the former Helsinki University of Technology~\cite{nissila2002instrumentation, nissila2005instrumentation}. 

\subsubsection{Prism-terminated optical fibers}

Fig.~\ref{fig:detector_prism_NA}A visualizes the cross-section of a single prism terminal of the fiber-optic probe at Aalto University, and the central and outermost photon trajectories that can lead to detection. The trajectories are depicted as red dashed lines. Identical components are used to guide light into the tissue. The fiber bundles lead the photons to photomultiplier tube (PMT) detectors.

\begin{figure}[b!]
\centering
\includegraphics[width=\textwidth]{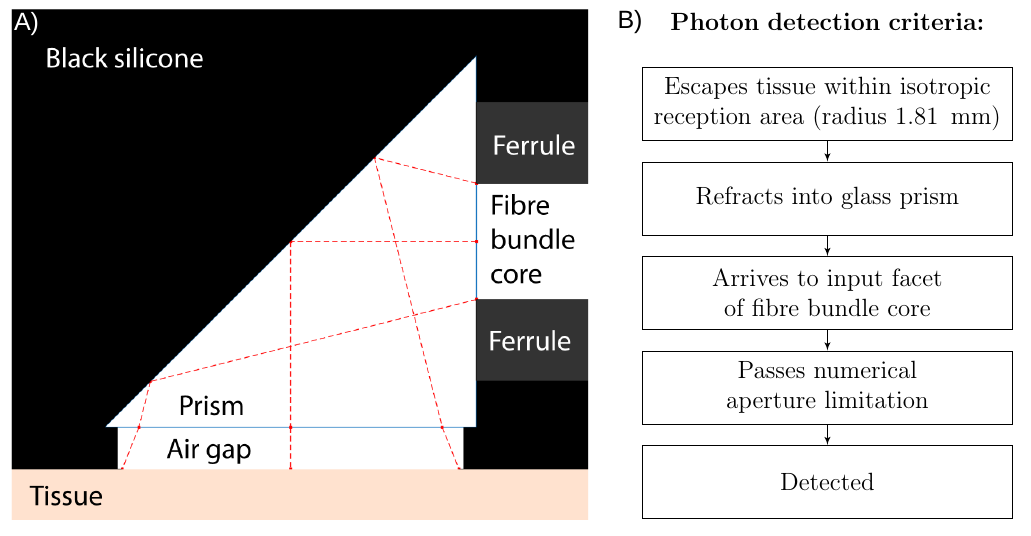} 
\caption{A) The glass prism terminal that collects light exiting the tissue and guides it to the detector fiber bundle core, viewed from the side. B) The test criteria that implement the detection system as a post-processing step following the Monte Carlo simulations.}
\label{fig:detector_prism_NA}
\end{figure}

The right-angle glass prism has an edge length of 4~mm and a refractive index of 1.5168. The optical fiber bundle core has a diameter of 1.25~mm and consists of individual fibers with numerical aperture (NA) of 0.37. The air gap is assumed to be 0.45~mm thick with refractive index of 1.0. The purpose of the air gap is to allow the edges of the silicone support to secure the prism in place while preventing the prism corners from contacting and potentially damaging the scalp. At the source sites, the prism also slightly broadens the light beam, thereby reducing sensitivity to small inhomogeneities on the tissue surface. In addition, the air gap distributes the heat generated by light absorption under the source optode over a larger area, potentially enabling the safe use of higher source power.  

Advantages of the design include multimodal compatibility, potentially high optode density, and good resistance to motion artifacts. A disadvantage is that individual optodes cannot be easily removed once the probe has been cast. Increasing the NA and overall diameter of the core of the detector fiber bundles would improve photon collection efficiency, but may require a corresponding adjustment of the prism size. The construction of a slightly modified design of the fiber-optic probe with entirely non-magnetic components and 3~mm prisms was described in~\cite{autti2025simultaneously}, where the probe was used for simultaneous high-density OT and magnetoencephalography (MEG) in adults. 

\subsubsection{Modeling detection in prism-terminated optical fiber}

For this work, we programmed a function that models the detector optode consisting of prism terminal and optical fiber bundle via a post-processing step that selects the detected photon packets. The criteria that the function tests are listed in Fig.~\ref{fig:detector_prism_NA}B. Only the detected photon packets are replayed when computing the Jacobians in MCX.

The starting point was set as isotropic reception on the head surface within radius of 1.81~mm from the detector position. In MCX, by default all photon packets that exit the tissue within this region are detected, and the post-processing function takes the exit positions and directions of these photon packets as its input. The reception area was selected as the region from which at least some photons can propagate to the fiber bundle via the prism.  

Snell's law is applied to model changes in photon direction at the air--glass interface. A glass square with 4~mm edges is used instead of the prism, as it produces an equivalent result in this geometry due to the law of reflection and symmetry, provided that no scattering or absorption occurs at the hypotenuse. Photons that arrive at the input facet of the detector fiber bundle with an angle of incidence within the NA-limits are detected. 

In the results, the data according to prism-terminated optical fiber detectors is compared to isotropic reception area with radius of $1.81$~mm, i.e., the starting point in the post-processing, which we have also used in combination with real measurements as a compromise to increase SNR in~\cite{maria2022imaging,shekhar2024maternal}. In addition, an isotropic detector with reception area corresponding to the diameter of the optical fiber bundle without the prism terminal is considered.

Cubic voxels with dimensions $0.5 \times 0.5 \times 0.5$~mm$^3$ were used to distinguish superficial phenomena at a higher resolution. Since the limitations of the DA can be ignored when using Monte Carlo, separate optical parameters were assigned for CSF in the subarachnoid region versus less diffusive CSF in the sulci and deeper ventricles~\cite{maria2022imaging, hirvi2023effects}. The sources were modeled as collimated Gaussian beams with waist radius of $1.25$~mm, approximating the light spread at the prism and air gap~\cite{maria2022imaging,shekhar2024maternal}. Reflections at incidence were neglected, although in the real-world implementation of the probe, an air gap exists also between the source prism terminal and the skin. An air gap was modeled between the probe and the skin, allowing light to reflect back into the head at the interface, although in practice the probe is pressed as firmly against the head. Alternatively, the whole probe could be incorporated to the head model, allowing internal reflections only at optode air gaps and outside of the area of the scalp covered by the probe.

\subsubsection{Split-voxel boundary}
\label{sec:svmc}
In the model for the prism-terminated optical fibers, the air gap width (0.45~mm) is the perpendicular (i.e., along the outward normal direction of tissue) distance from the center of the detection area to the center of the lower surface of the glass prism. The detector model adapts the air gap width to account for non-flat surfaces, but a sufficiently smooth exterior boundary is required to justify representing the detection area with a common surface normal direction. Furthermore, the model assumes that the tissue--air interface is also sufficiently convex such that photon packets exiting the tissue cannot re-enter it. A more precise model allowing re-entry would require extending the tissue model to include the air gap. 

We compared the different detector models in a flattened voxel model, and in a planar quad mesh--voxel hybrid model using the split-voxel Monte Carlo (SVMC) approach~\cite{yan2020hybrid}. SVMC improves the modeling of curved boundaries by assigning a tilted plane to define the tissue division in each boundary voxel. Each split voxel then contains two subvolumes and piecewise constant absorption and scattering coefficients. In actual application use, it would be favorable to consider even smoother surface meshes, or meshes composed of elements that fully cover the detector area. Hybrid or fully meshed models can be considered. 

The SVMC tissue boundaries do not align with the voxel boundaries, thus we had to reconsider the definition for the per-voxel derivatives in the split voxels. Since the SVMC outputs follow the voxel grid, and coefficient changes can eventually only be reconstructed as constant for each voxel, it was meaningful to consider the per-voxel coefficient changes as constant. Thus, following Sec.~\ref{sec:pMC_generalized}, the Jacobians can be counted according to the MC counterparts for the theoretical formulas in Eqs.~\eqref{eq:general_probability_absorption} and \eqref{eq:general_probability_scattering}. Furthermore, the SVMC replay mode can aggregate the per-voxel total path lengths and scattering counts similarly to the voxel-based MCX, except that the weighted scattering counts in Eq.~\eqref{eq:MCXimplementation} must be divided by the current scattering coefficient during the replay in SVMC to implement Eq.~\eqref{eq:general_probability_scattering}. In voxel-based MCX, this can and has been left as a manual post-processing task. The difference in handling the division by the scattering coefficient was the only update required to the v2025.10 release of MCX to extend the replay mode to SVMC.

\section{Results}
\label{sec:results}

\subsection{Validation of frequency- and time-domain Jacobians}

The three newly derived MCX output quantities (\#1 -- \#3), required for all FD/TD Jacobian terms shown in Eq.~\eqref{eq:MCXimplementation}, are essentially weighted sums over detected photon packets. As a result, we perform the initial verification at the photon packet level. First, we picked one detected photon packet, and compared the new replay outputs to per-voxel values computed manually using only the previously established replay types. Specifically, pure voxel-wise path lengths and scattering counts were obtained by setting absorption to zero, thereby isolating the geometric photon-packet trajectories. Following this, the obtained per-voxel values were weighted by the explicitly computed detected packet weight, TOF and the cosine- and sine-terms, depending on the replay type.

This validation was then extended to multiple detected photon packets by manually summing contributions over 2, 4, 40, 60, and 201 packets and comparing the results against the corresponding replay outputs. In all cases, the per-voxel relative differences between manually computed and automatically accumulated values were either exactly zero or within the numerical precision expected for single-precision data in MATLAB. Additional consistency checks confirmed that the replay-based Jacobian computation is consistent for different voxel side lengths (0.5~mm and 1~mm), and that the FD outputs behave as expected when setting the modulation frequency to zero, i.e., the imaginary parts vanish. These tests establish the correctness of the MCX implementations. Following these sanity checks, Jacobians computed using MCX were compared against FD Jacobians obtained by solving the DA with FEM.

Figures \ref{fig:N41_det2}--\ref{fig:N172_det3} present axial slices of TD and FD Jacobians computed in the head models of two neonates with voxel edge length of 1\,mm. Figures~\ref{fig:N172_det2} and \ref{fig:N172_det3} represent the same full term neonate; Fig.~\ref{fig:N41_det2} represents a preterm. The selected SDSs are 37\,mm in Fig.~\ref{fig:N41_det2}, 12\,mm in Fig.~\ref{fig:N172_det2}, and 24\,mm in Fig.~\ref{fig:N172_det3}. For the FD simulations, the light source was modulated at a radio frequency of $f = 100$\,MHz; hence, the TD mean TOF Jacobians are scaled by the corresponding angular modulation frequency $\omega=2\pi f$ to enable comparison with the FD phase sensitivities. All Jacobians were simulated using $10^{13}$ photon packets to ensure sufficient SNR for the scattering Jacobians.

The Figs.\ \ref{fig:N41_det2}--\ref{fig:N172_det3} share a common layout: panels (A) and (B) show the head model and the axial slice in the source--detector center plane, respectively, while panels (C--N) display Jacobians computed with MCX and with the DA. Columns 1--3 correspond to TD MCX Jacobians, FD MCX Jacobians, and FD diffusion-approximation Jacobians, respectively, whereas the different row blocks show Jacobians with respect to absorption (rows 1--2) and scattering (rows 3--4).

The computed Jacobians exhibit the characteristic banana-shaped sensitivity profiles connecting the source and the detector~\cite{hillman2007optical, durduran2010diffuse}. The scattering Jacobians appear approximately two orders of magnitude smaller than their absorption counterparts, as expected based on~\cite{yao2018direct}. The FD and TD Jacobians computed with MC are visually very similar in most cases. For the short source–detector separation of 12\,mm in Fig.~\ref{fig:N172_det2}, small differences can be observed between the scaled mean TOF and phase-shift scattering Jacobians at the outermost voxels near the optodes. Nevertheless, these observations support the theory in Sec.~\ref{sec:FD_vs_TD} on internal consistency of the TD and FD Jacobians at the considered low intensity modulation frequency.

When considering the Jacobians computed using the DA, the overall sensitivity profiles exhibit similar spatial shapes to their MC counterparts. However, there are some clearly observable differences in the derivative values. Specifically, the sensitivity profiles computed with the DA tend to predict higher sensitivities than those computed by MC. This is particularly the case for the scattering Jacobians and the absorption Jacobian of the phase shift, whereas the absorption Jacobian of the log-amplitude shows good agreement with its MC counterpart. Similar differences can be observed in the MC versus DA Jacobians for the real and imaginary components of the measurements ($\hat{X}$ and $\hat{Y}$, Eqs.~\eqref{eq:MC_X_and_y_mua_derivatives}--\eqref{eq:MC_X_and_y_mus_derivatives}; not visualized), although the matching overall magnitudes suggest that the source (strength) models in MC and DA correspond to each other well.

The higher sensitivities predicted by the DA are predominantly localized around the subarachnoid CSF layer in all tested cases. As is well known, the DA is not valid in low-scattering media such as the CSF. To investigate whether this limitation explains the observed discrepancies, a modified head model was created in which the CSF layer was replaced by grey matter optical properties, while all other model parameters were kept identical to those used in Fig.~\ref{fig:N172_det3}. A similar strategy has been used previously to assess the influence of the CSF layer~\cite{heiskala2007significance}. This model no longer represents a realistic human head, but instead corresponds to a multi-layer, fully diffusive and turbid medium for which the DA is justified.

The FD Jacobians computed with MC and with the DA in this fully diffusive model are shown in Fig.~\ref{fig:N172_det3_no_CSF}. In addition to axial slice visualizations, the figure shows depth-resolved sensitivity profiles obtained by aggregating Jacobian values within 1-mm-thick disk-shaped cross sections along an 18-mm–radius cylinder. The values for each cylinder are normalized with the respective maximum absolute disk value. The cylinder is positioned midway between the source and detector and extends into the head along the surface normal. The sensitivity profiles computed with MC and the DA show close agreement, both in the disk-averaged profiles and in the axial slices. This result validates our MC FD Jacobian implementation and attributes the differences in Figs.~\ref{fig:N41_det2}--\ref{fig:N172_det3} to the known limitations of the DA. A qualitative comparison of Figs.~\ref{fig:N172_det3} and \ref{fig:N172_det3_no_CSF} further indicates that the presence of CSF increases sensitivity to tissues adjacent to the CSF, consistent with previous observations~\cite{heiskala2007significance}. The increased sensitivity in the presence of low-scattering tissue can be explained by noting that when the voxel-wise scattering coefficient approaches zero, the derivatives in Eqs.~\eqref{eq:MC_probability_scattering}, \eqref{eq:MC_X_and_y_mus_derivatives}, \eqref{eq:meantof_mus_derivative}, and \eqref{eq:MCXimplementation} reduce to absorption derivatives, which are generally 1--2 orders of magnitude larger than scattering derivatives.

\begin{figure}[h!]
\centering 
\includegraphics[width=\textwidth]{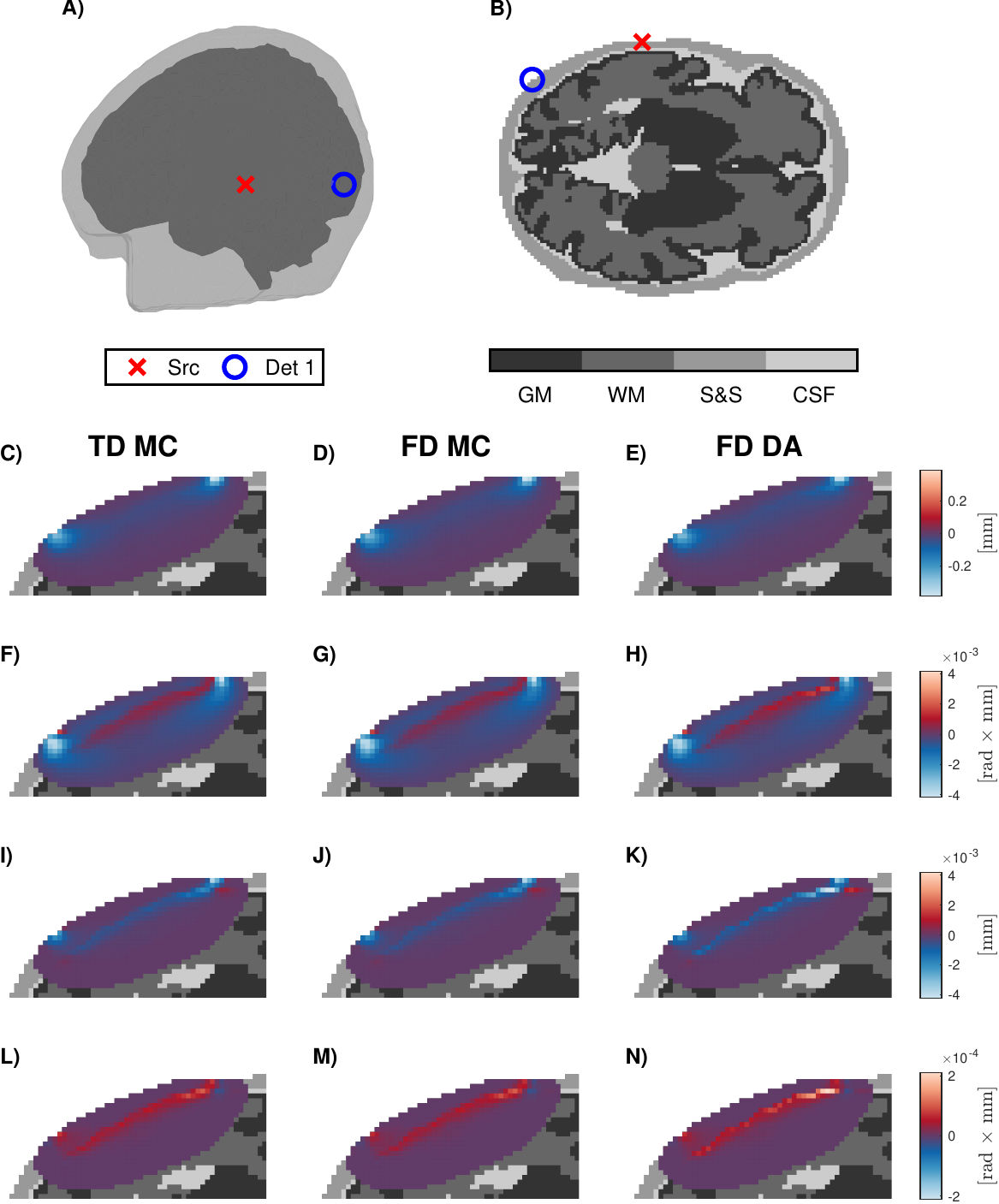}
\caption{Axial slices of time-domain (TD) and frequency-domain (FD) Jacobians in the voxel brain model of a preterm neonate.
(A) Brain model and (B) axial slice in the source--detector plane showing tissue segmentation. Panels (C–N) show Jacobians for a single source–detector pair (detector 1) separated by 37\,mm. Columns correspond to TD Jacobians computed with MCX, FD Jacobians computed with MCX, and FD Jacobians computed using the diffusion approximation (DA) with the finite element method. Rows show Jacobians with respect to voxel-wise absorption (rows 1--2) and scattering coefficients (rows 3--4). The first and third rows correspond to TD log-integrated intensity and FD log-amplitude, while the second and fourth rows show TD scaled mean time-of-flight and FD phase-shift Jacobians.}
\label{fig:N41_det2}
\end{figure}

\begin{figure}[h!]
\centering 
\includegraphics[width=\textwidth]{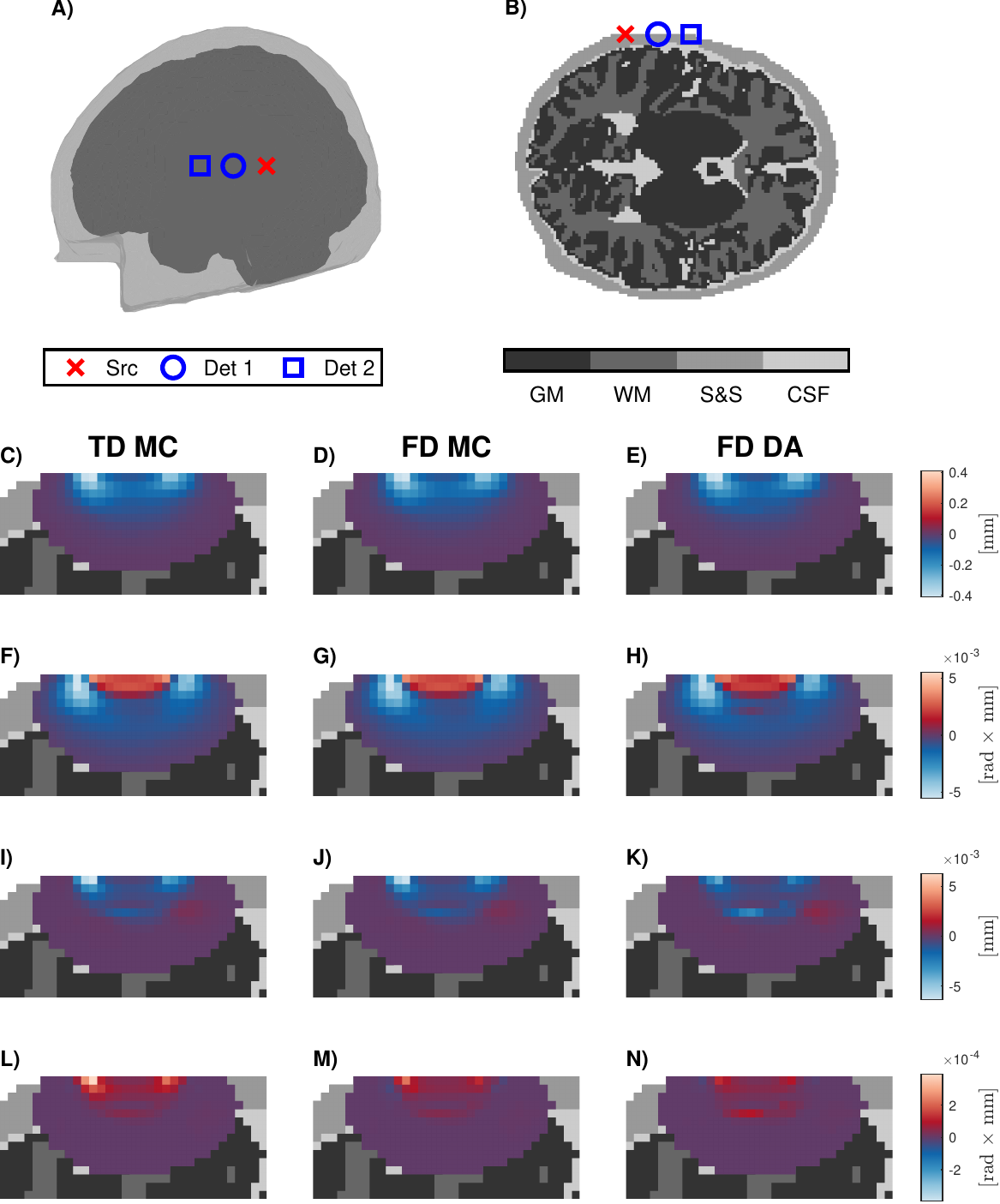}
\caption{Axial slices of time-domain (TD) and frequency-domain (FD) Jacobians in the voxel brain model of a full-term neonate.
(A) Brain model and (B) axial slice in the source--detector plane showing tissue segmentation. Panels (C–N) show Jacobians for a single source–detector pair (detector 1) separated by 12\,mm. Columns correspond to TD Jacobians computed with MCX, FD Jacobians computed with MCX, and FD Jacobians computed using the diffusion approximation (DA) with the finite element method. Rows show Jacobians with respect to voxel-wise absorption (rows 1--2) and scattering coefficients (rows 3--4). The first and third rows correspond to TD log-integrated intensity and FD log-amplitude, while the second and fourth rows show TD scaled mean time-of-flight and FD phase-shift Jacobians.}
\label{fig:N172_det2}
\end{figure}

\begin{figure}[h!]
\centering 
\includegraphics[width=\textwidth]{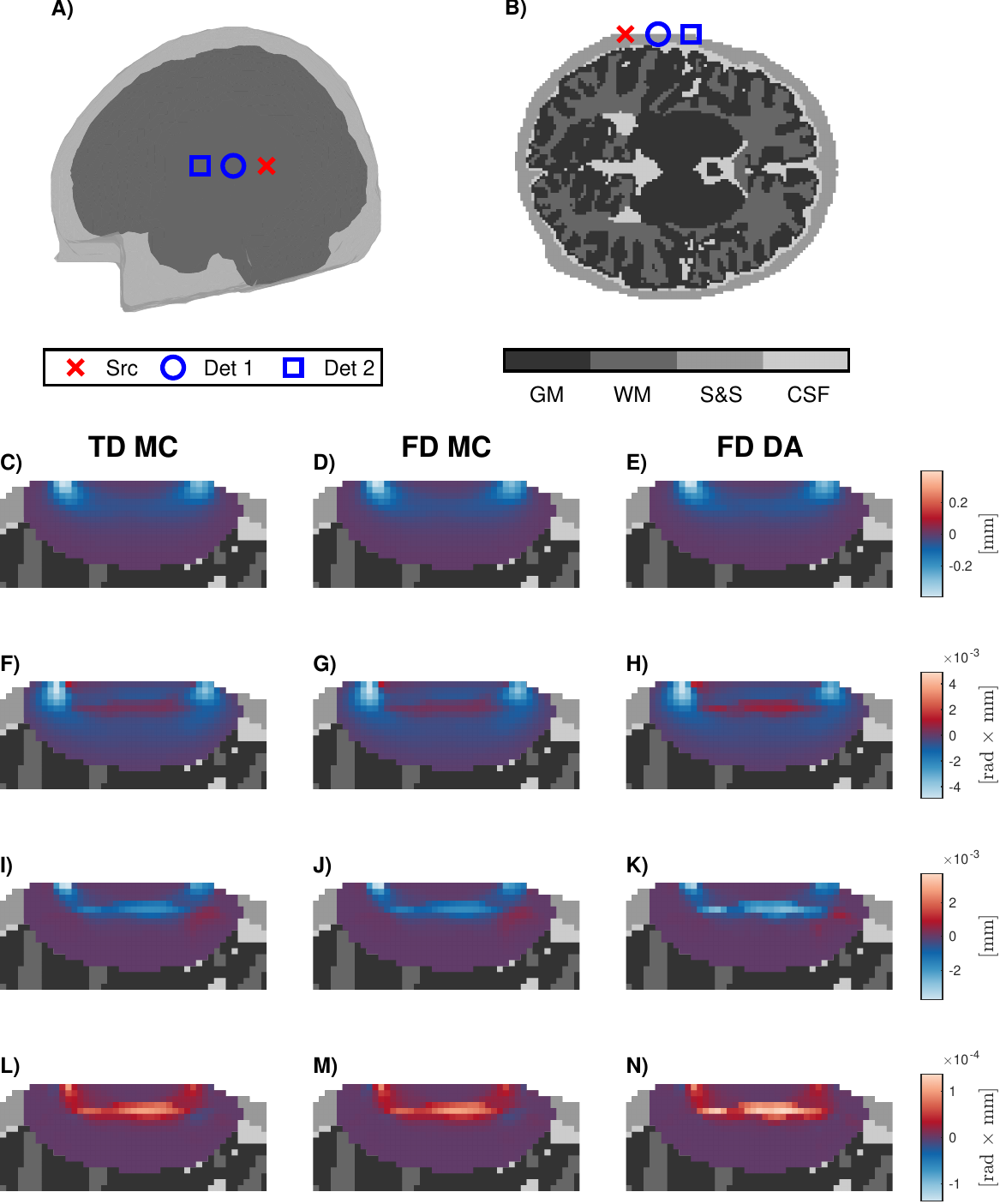}
\caption{Axial slices of time-domain (TD) and frequency-domain (FD) Jacobians in the voxel brain model of a full-term neonate.
(A) Brain model and (B) axial slice in the source--detector plane showing tissue segmentation. Panels (C–N) show Jacobians for a single source–detector pair (detector 2) separated by 24\,mm. Columns correspond to TD Jacobians computed with MCX, FD Jacobians computed with MCX, and FD Jacobians computed using the diffusion approximation (DA) with the finite element method. Rows show Jacobians with respect to voxel-wise absorption (rows 1--2) and scattering coefficients (rows 3--4). The first and third rows correspond to TD log-integrated intensity and FD log-amplitude, while the second and fourth rows show TD scaled mean time-of-flight and FD phase-shift Jacobians.}
\label{fig:N172_det3}
\end{figure}

\begin{figure}[h!]
\centering 
\includegraphics[width=\textwidth]{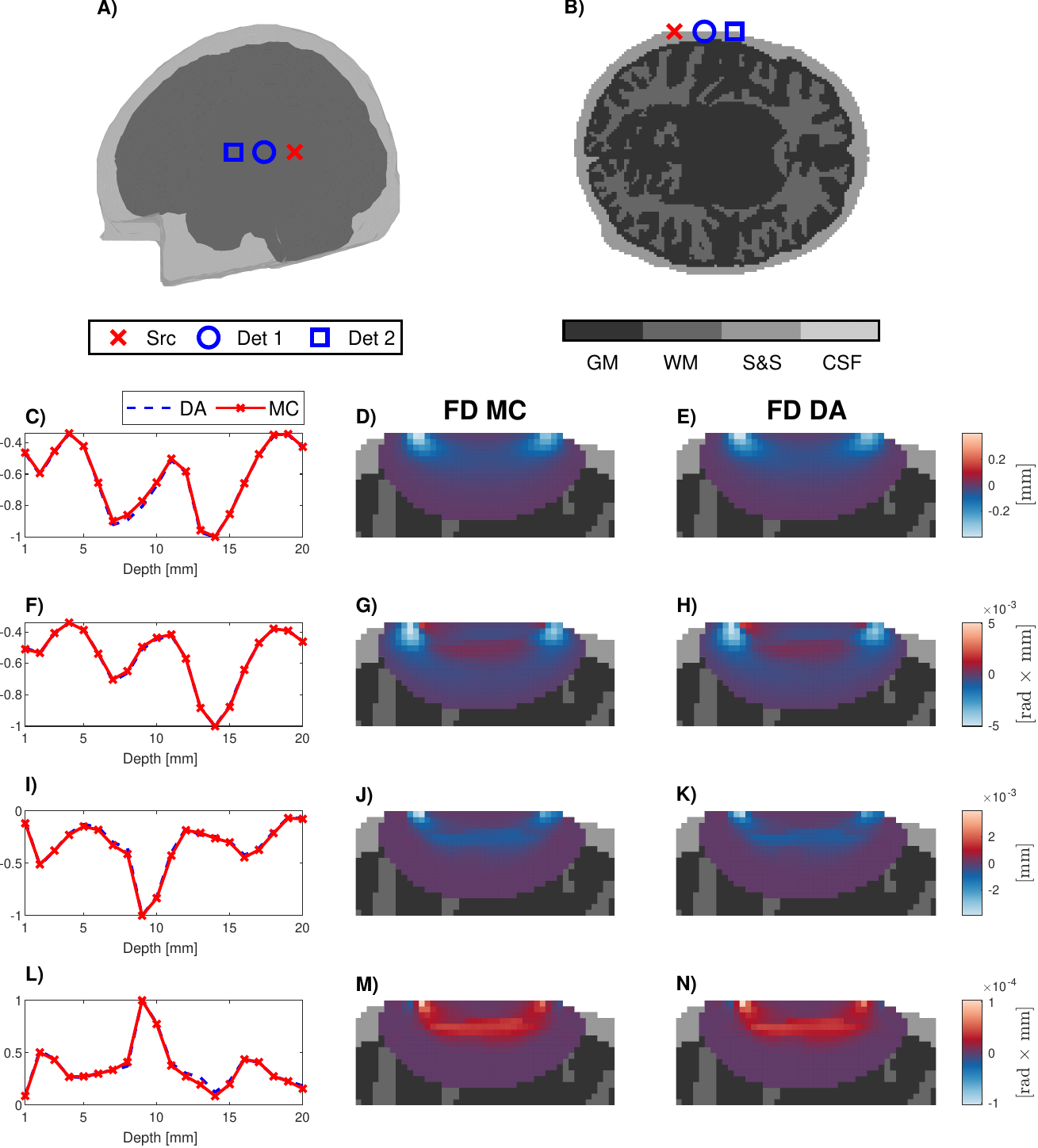}
\caption{Frequency-domain (FD) Jacobians in a fully diffusive model.
(A) Modeled domain and (B) axial slice in the source--detector plane showing tissue segmentation. Panels (C, F, I, L) show depth-resolved sensitivity profiles, computed by aggregating voxel-wise Jacobian values within 1-mm-thick disks of radius 18\,mm (0.75 $\times$ SDS) along a rectangular cylinder drilled into the head midway between the source and detector. The resulting values are normalized by the maximum absolute Jacobian value. Panels (D, E, G, H, J, K, M, N) show axial slices of the Jacobians computed with MCX (second column) and using the diffusion approximation (DA) with the finite element method (third column). Rows show Jacobians with respect to voxel-wise absorption (rows 1--2) and scattering coefficients (rows 3--4). The first and third rows correspond to log-amplitude, while the second and fourth rows show phase-shift Jacobians. Compared to Fig.~\ref{fig:N172_det3}, the cerebrospinal fluid layer was assigned grey matter optical properties, yielding a fully diffusive multi-layer model in which the DA is expected to be valid.}
\label{fig:N172_det3_no_CSF}
\end{figure}

\clearpage
\subsection{Detector and boundary models in Monte Carlo}

Figs.~\ref{fig:DetectorModels_Idea}--\ref{fig:SVMC_DetectorModels_Jac_Slices} illustrate the effects of the refined detector model on photon exit positions and directions, absolute data, and sensitivity profiles compared with isotropic reception. The model accounts for the prism terminal and the NA limitation at the optical fiber. All figures correspond to the same axial slice of the preterm neonatal head model, also used in Fig.~\ref{fig:N41_det2}. Figs.~\ref{fig:DetectorModels_Idea}--\ref{fig:DetectorModels_Jac_Slices_4mm} include one source and five detectors with SDSs of 4, 6, 10, 14, and 18~mm in a manually flattened region, whereas Fig.~\ref{fig:SVMC_DetectorModels_Jac_Slices} contains a single source--detector pair with SDS of 5.7~mm in a curved, terraced region. Relatively short SDSs and a voxel side length of 0.5\,mm are employed throughout.

Fig.~\ref{fig:DetectorModels_Idea} shows the exit positions and directions for five randomly selected photon packets per detector that exit within the isotropic detection area, and either A) continue to the optical fiber or B) are rejected. Photon packets that fail to reach the optical fiber are observed to typically exit either at relatively large angles with respect to the surface normal or near the boundary of the isotropic reception area.

\begin{figure}[h!]
\centering
\includegraphics[width=\textwidth]{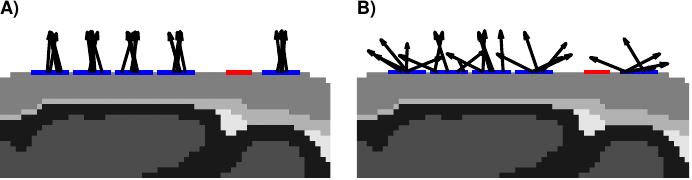}
\caption{Exit positions and directions (black arrows) for photon packets that pass the isotropic reception area with radius of 1.81\,mm and are either A) accepted or B) rejected by the prism and numerical aperture limitation. The red line indicates the source region, and the blue lines indicate detector regions for source--detector separations of 4, 6, 10, 14, and 18~mm.}
\label{fig:DetectorModels_Idea}
\end{figure}

Figs.~\ref{fig:DetectorModels_AbsData}A and B show the simulated FD absolute data as a function of SDS for four different detector models. The log-amplitude values are normalized to unit source strength. Black color corresponds to the isotropic reception area with radius of 1.81\,mm, which is the initial step in the refined detection procedure in Fig.~\ref{fig:detector_prism_NA}B. Red color marks the simulated measurements according to the prism-terminated optical fiber, i.e., photon packets that pass the selection procedure. Rejection of photon packets reduces the detected amplitudes approximately by 98\%, which agrees with the photon packet rejection ratio. If the lowest log-amplitude values are aligned, both detector models produce similar trends, although isotropic reception yields slightly higher values at the shortest SDS (not visualized). Photon TOFs are marginally increased due to the restricted detection.

To discover the step in the detection procedure in Fig.~\ref{fig:detector_prism_NA}B that rejects most photon packets, Fig.~\ref{fig:DetectorModels_AbsData} also shows the measurements obtained by isotropic reception over the cross-sectional area of the optical fiber bundle (blue) and the NA-limiting optical fiber without the glass prism between the air gap and the fiber bundle (pink; overlaps with red). The latter corresponds to a realistic alternative probe design, where the fibers are mounted directly into the probe. The glass prism yields moderately higher maximum amplitude values, presumably because it increases the gap between the head surface and the fiber, which widens the acceptance area on the head surface. Clearly, mainly the reception area, the diameter of the fiber bundle, and the NA-limitation contribute to the amplitude loss. Furthermore, setting the isotropic reception area according to the diameter of the fiber bundle gives a more realistic estimate of the NA-limitation, since the original larger reception area simply estimates the exit region from which photons could end up to the fiber bundle. This should also be noted when comparing the sensitivity profiles for the large isotropic reception area versus the refined detection procedure. 

\begin{figure}[t!]
\centering
\includegraphics[width=0.9\textwidth]{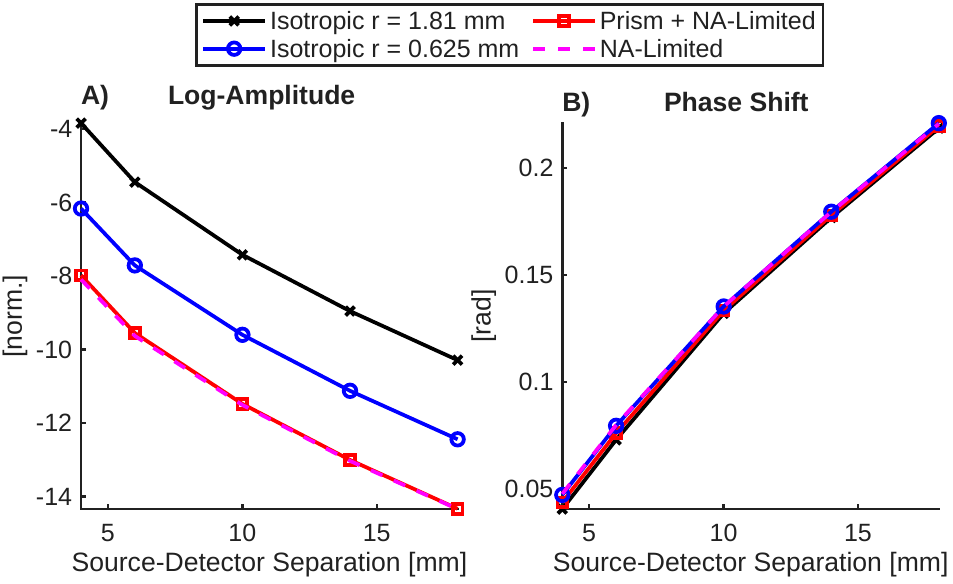}
\caption{Simulated frequency-domain A) log-amplitude and B) phase data according to four different detector models in Monte Carlo.} 
\label{fig:DetectorModels_AbsData}
\end{figure}

Fig.~\ref{fig:DetectorModels_Jac_CylinderDrill} visualizes depth-resolved sensitivity profiles (as in Fig.~\ref{fig:N172_det3_no_CSF}) for all eight Jacobian types and the two detector models for SDSs of 4, 6, and 10~mm and $5\times 10^{11}$ simulated photon packets. The cylinder cross sections are 4-millimeter high rectangles with width corresponding to half of the respective SDS. The values are normalized with the maximum absolute value along the cylinder. Black color corresponds to the isotropic reception area with radius 1.81\,mm, and red color is restricted to the photon packets that pass the prism-terminated optical fiber detection procedure in Fig.~\ref{fig:detector_prism_NA}B. Figs.~\ref{fig:DetectorModels_Jac_Slices_4mm} shows the corresponding axial slices at the source--detector level for the shortest channel with an SDS of 4~mm. The prism-terminated optical fiber detectors reduce the relative sensitivity at the most superficial tissues and slightly increase it deeper in the head. The CSF layer at 3~mm depth can be clearly seen in the scattering Jacobians for the largest SDS, where sensitivity is less focused on the extracerebral tissue.

\begin{figure}[b!]
\centering 
\includegraphics[width=\textwidth]{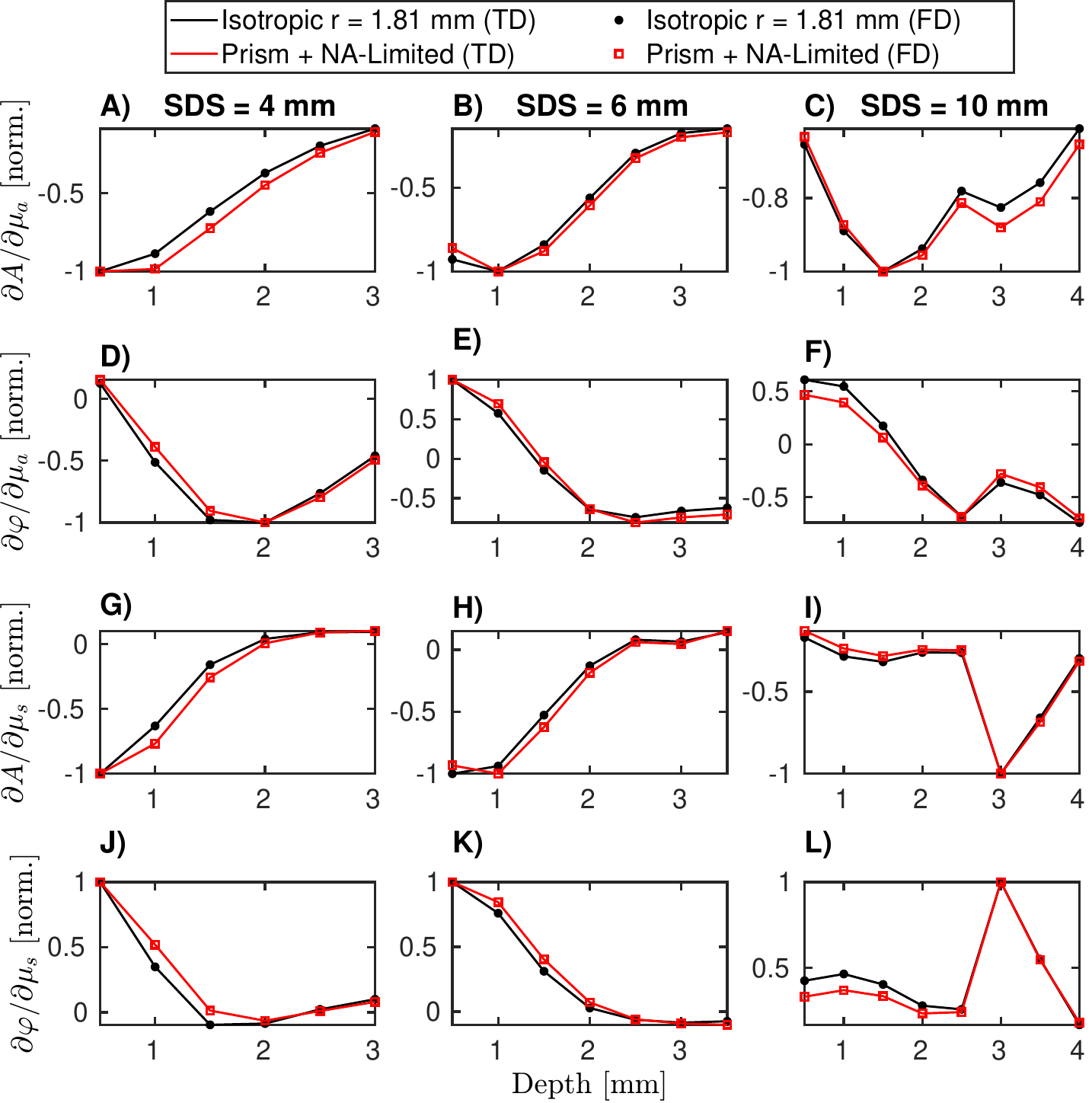} 
\caption{Jacobians for reception area radius of 1.81\,mm with isotropic (black) versus numerical aperture (NA) limited detection in prism-terminated optical fibers (red), evaluated over 0.5-millimeter thick intersections along a rectangular cylinder drilled into the head half-way between the respective source--detector pair. The source--detector separations (SDSs) are 4, 6, and 10\,mm for columns 1, 2, and 3, respectively. The cylinder width is half the corresponding SDS, and the height is 4\,mm. First row: intensity and amplitude versus absorption. Second row: mean time-of-flight (TOF) and phase versus absorption. Third row: intensity and amplitude versus scattering. Fourth row: mean TOF and phase versus scattering. Only the FD data type is labeled on the vertical axis for clarity. TD = time-domain. FD = frequency-domain}
\label{fig:DetectorModels_Jac_CylinderDrill}
\end{figure}

\begin{figure}[h!]
\centering
\includegraphics[width=\textwidth]{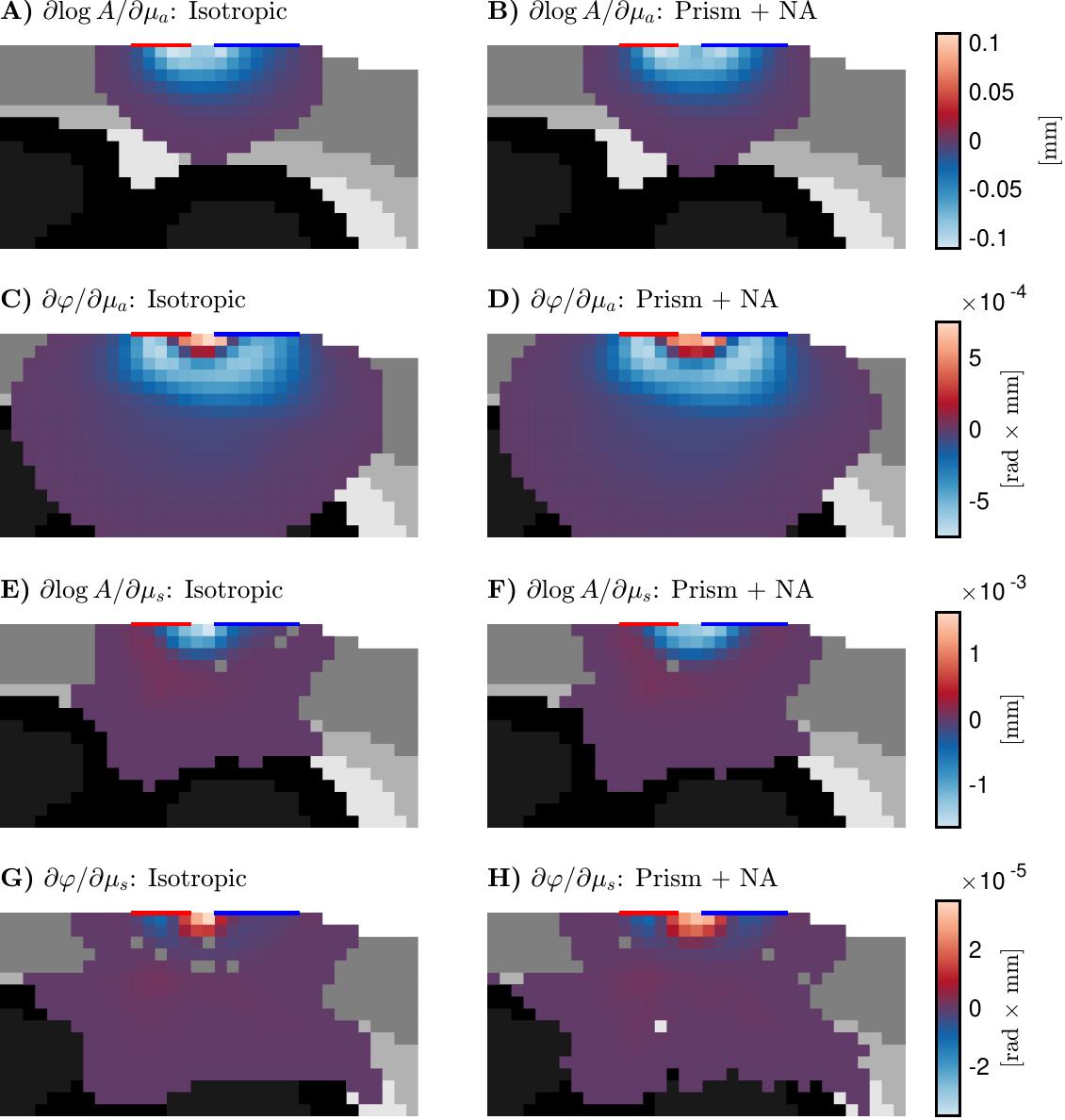} 
\caption{Axial slices of frequency-domain Jacobians of log-amplitude (A, B, E, F) and phase (C, D, G, H) with respect to absorption (A--D) and scattering (E--H) computed for a single source--detector pair 4.0\,mm apart with $5\times 10^{11}$ simulated photon packets in MCX. (A, C, E, G) Jacobians obtained with an isotropic detection area with radius of 1.81\,mm; (B, D, F, H) Jacobians obtained with the numerical-aperture-limiting model for prism-terminated optical fiber detectors. The visualization threshold is 0.5\% (A--D) or 0.2\% (E--H) of the maximum absolute value in each row. The red and blue lines indicate the source and detector regions, respectively. The voxel side length is 0.5\,mm; accordingly, the magnitudes are approximately eight times smaller than at 1-millimeter resolution.}
\label{fig:DetectorModels_Jac_Slices_4mm}
\end{figure}

Finally, Fig.~\ref{fig:SVMC_DetectorModels_Jac_Slices} compares the FD Jacobians for a short channel with an SDS of 5.7\,mm obtained using isotropic reception in a voxel model versus NA-limited detection by the prism-terminated optical fiber in the hybrid model. We simulated $3.5 \times 10^{12}$ photon packets for the limited detection model, whereas $2 \times 10^{10}$ was sufficient for the isotropic model at the considered short SDS. The plane normals of split boundary voxels deviate by 1.2--7.8$\degree$ (median 3.9$\degree$, mean 4.4$\degree$, standard deviation 1.7$\degree$) from the selected mean normal direction in the detection area. The NA limitation lowers sensitivity in the most superficial voxels, similarly to Fig.~\ref{fig:DetectorModels_Jac_Slices_4mm} for the fully voxelated models. Otherwise, the similarity of the sensitivity profiles and magnitudes indicates consistency of the Jacobian simulations in the SVMC mode. However, the scattering Jacobians exhibit suspiciously large values in the voxels that overlap with the source, suggesting that the path length contribution is too small relative to the scattering count, although it would seem more likely that the latter could be too small in a split $0.5 \times 0.5 \times 0.5$~mm$^3$ voxel. Some blurring is introduced when the SVMC grid is shifted to match the voxel grid. 

\begin{figure}[h!]
\centering
\includegraphics[width=\textwidth]{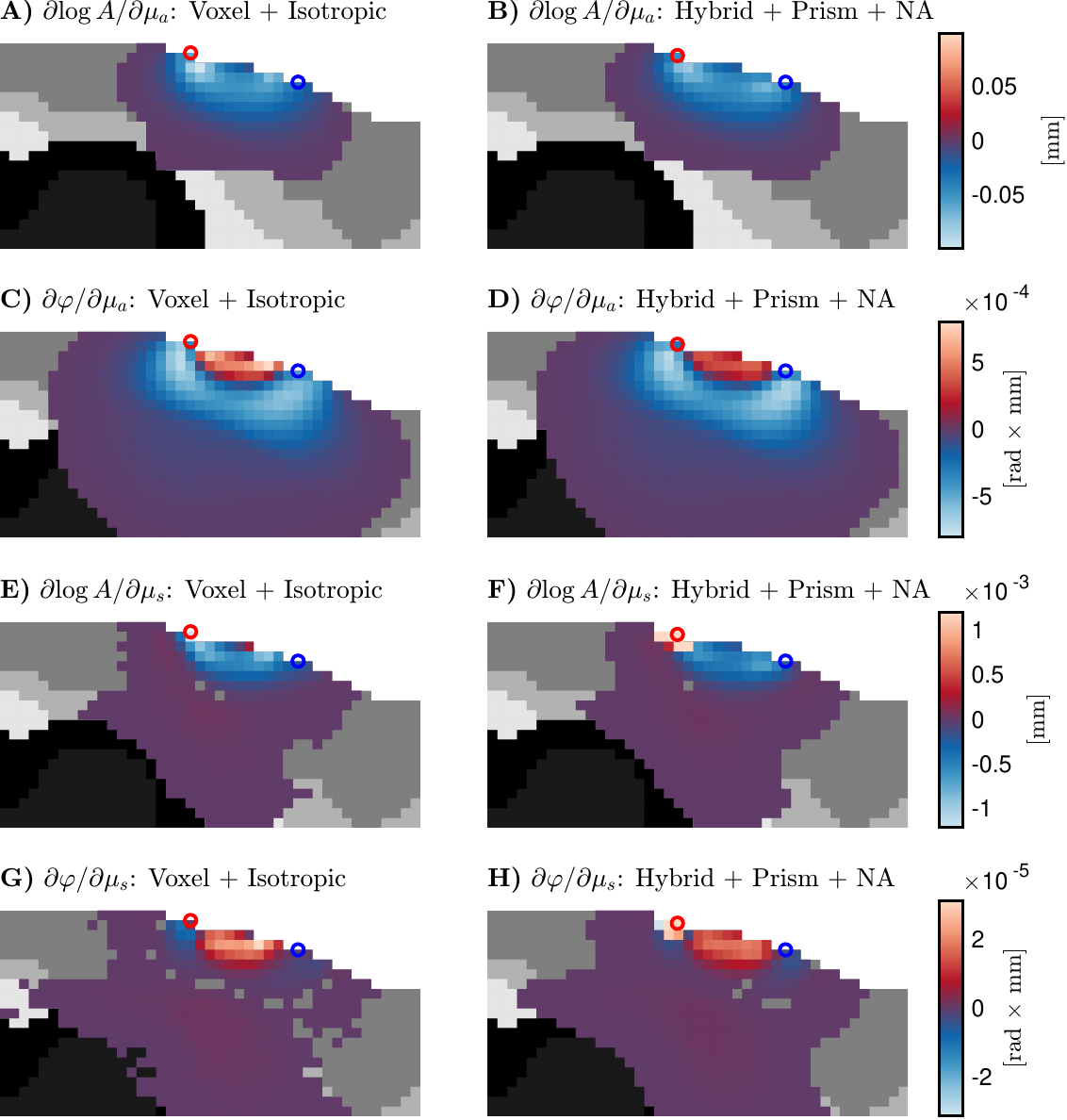}
\caption{Axial slices of frequency-domain log-amplitude (A, B, E, F) and phase (C, D, G, H) Jacobians with respect to absorption (A--D) and scattering (E--H), computed for a single source–detector pair 5.7~mm apart. The first column (A, C, E, G) considers a voxel model with $2 \times 10^{10}$ simulated photon packets and an isotropic reception profile. The second column (B, D, F, H) presents a split-voxel hybrid model with $3.5 \times 10^{12}$ simulated photons packets and prism-terminated optical fiber limiting the detection rate based on the numerical aperture (NA). The split boundary voxels' plane normals are at a 1.2--7.8$\degree$ angle with respect to the selected mean normal direction in the detection area. The voxel side length is 0.5\,mm; accordingly, the magnitudes are approximately eight times smaller than at 1-millimeter resolution. Visualization threshold is set to 0.5\% (A--D) or 0.2\% (E--H) of absolute maximum in each row, which also defines the colorbar limits.} 
\label{fig:SVMC_DetectorModels_Jac_Slices}
\end{figure}

\clearpage
\section{Discussion}
\label{sec:discussion}

Optical tomography (OT) utilizes visible red and near-infrared light for three-dimensional imaging. Frequency-domain (FD) and time-domain (TD) optical imaging systems provide time-resolved information and multiple measurement types with different sensitivity profiles, which can be utilized to improve image quality. Monte Carlo (MC) methods are typically considered the gold standard for computing the sensitivity profiles for the measurements, or the Jacobians. Monte Carlo eXtreme (MCX) is a widely used open-source software for efficient parallel simulation of photon packets. 

In this work, we derived formulas for directly computing sensitivity profiles in MC for the FD log-amplitude and phase-shift measurements, in addition to the TD intensity and mean time-of-flight (TOF) measurements, with respect to the piecewise constant absorption and scattering coefficients in the domain. FD scattering Jacobians, as well as mean TOF absorption and scattering Jacobians, were implemented to complement previously published Jacobian types in MCX. A detector function was developed to select photon packets that can realistically be detected based on their exit positions and directions at the tissue boundary, consistent with imaging probes using prism terminals and optical fibers. In addition to voxelated models, planar quad mesh--voxel hybrid models with split boundary voxels were considered for modeling the curvature of the boundary more accurately.

A common alternative to MC simulations is to use the finite element method (FEM) to numerically solve the diffusion approximation (DA) to the radiative transfer equation. The new MC-based Jacobian types were compared against DA Jacobians computed using an in-house FEM solver developed at Aalto University. The Jacobians obtained from MC and FEM/DA showed close agreement, particularly when the less diffusive CSF in the model was replaced by fully diffusive tissue. The effect of different anisotropy factors on the efficiency of computation of the Jacobians using Monte Carlo is a subject of further study.  

All the examples presented in this work consider neonatal head models segmented into four or five tissue types. These correspond to complex heterogeneous structures, but the MC software is freely available to be used with any medium where the RTE is a suitable model for light transport. The added sensitivity modalities can also be extended to compute Jacobians for higher temporal moments, such as the variance of the TPSF~\cite{hillman2000calibration} or additional FD data types~\cite{sassaroli2023novel}.

In functional brain imaging with OT, hemodynamic changes are typically assumed to primarily affect the absorption coefficient. Consequently, changes in the scattering coefficient are often not reconstructed, and the corresponding scattering Jacobians are omitted from the linearized inverse problem. An interesting question is how scattering Jacobians could be used in practical neuroimaging applications. Neuronal activation induces vasodilation which could lead to measurable scattering changes. A faster optical signal which is possibly caused by scattering changes associated with electrophysiology has been investigated in animals~\cite{rector1997light} and humans ~\cite{steinbrink2005fast}. The imaging of neuronal activation can benefit from accurate modeling of the complex optical structure of the human head. If homogeneity within each tissue type can be assumed and anatomical prior information such as segmented MRI or atlas is available, the initial optical parameters of the most important tissue types can be estimated by fitting simulated data to calibrated measurements even without explicit image reconstruction~\cite{maria2022imaging}. In case the tissues are not modeled adequately by a model which assumes homogeneity within each tissue type, imaging of scattering and absorption using both Jacobians and nonlinear image reconstruction using regularization that respects the tissue boundaries could lead to improved accuracy. Given the observed close agreement between MC and DA simulations, one possible strategy would be to use DA to obtain a computationally efficient initial estimate, and then refine it using MC together with the improved detector model for source-detector pairs with the shortest separations as well as segmented non-diffusive regions to improve the accuracy of the estimated parameters. Incorporating a realistic detector model is expected to improve the agreement between simulated and measured absolute intensities and amplitudes, and may result in more robust estimates of the baseline optical parameters. In the future, MCX with the presented Jacobians could be used with nonlinear reconstruction methods incorporating anatomical priors to allow spatially varying baseline optical parameters while preserving the flexibility and accuracy of Monte Carlo methods.

In optical mammography~\cite{fang2009combined}, the location and oxygenation parameters of breast lesions could potentially be reconstructed from differences between the breast signal and a baseline measurement obtained after filling the breast cup with a diffusive, tissue-mimicking liquid. Beyond biomedical contexts, potential applications include environmental sensing, for example, detecting obstacles or inclusions in turbid media such as fog or water. Cross-talk between the optical parameters is commonly seen in images reconstructed based on amplitude and phase in FD (or correspondingly, mean time and integrated intensity in time-domain systems; e.g.~\cite{nissila2006comparison}), and should be considered a limitation of the approach. Cross-talk may be reduced by increasing the optode density.

A limitation of photon-replay-derived scattering Jacobians is the requirement of a large number of simulated photon packets compared to the absorption Jacobians, as also observed in~\cite{kangasniemi2024stochastic, amendola2024_part2}. One reason for this is that scattering Jacobians are computed as the difference between two replay outputs of relatively similar magnitude, as also discussed in~\cite{yamamoto2016frequency}. In this work, $\sim 10^{12}$--$10^{13}$ photon packets were required for visual smoothness in the scattering Jacobians. Stochastic optimization~\cite{kangasniemi2024stochastic} or deep learning methods could be used to denoise and improve the SNR. In highly scattering media, such as most biological tissues (with the exception of CSF, synovial fluid, and a few other low-scattering tissues), DA-based Jacobians can be computed significantly more efficiently. Regardless of whether the forward model is MC or DA, Jacobians can alternatively be computed using the adjoint method~\cite{yao2018direct}. The adjoint approach assumes reciprocity between sources and detectors~\cite{fang2004} and can utilize all the simulated photons to achieve significantly reduced stochastic noise in the Jacobians, whereas the replay method can only utilize detected photon packets, which typically constitute a very small fraction of all the simulated photon packets. When computing adjoint solutions by treating detectors as sources, the emission profile must be configured to match the physical light-collection characteristics of the detector. In addition, an appropriate scaling factor must be derived to ensure consistency of the adjoint Jacobians with the corresponding measurement definitions~\cite{yao2018direct}. A rigorous comparison between the replay and adjoint approaches is left for future work. Finally, a clear advantage of replay is that the Jacobians for multiple baseline absorption coefficients can be computed from the same trajectories due to mBLL, thus with no major increase in computation time~\cite{yao2018direct}.

The ability to accurately model light transport corresponding to short SDSs is a general advantage of MC methods. Including very short SDS channels in high-density probe designs can aid in separating extracerebral signals from brain activity and improve the interpretation of results in the context of imaging neuronal activity. Extracerebral physiology is one of the most significant counfounding factors in optical neuroimaging especially in the case of adult subjects with relatively thick extracerebral tissue layers.  Correct modeling of measurement sensitivity across tissues, even at short source--detector separations, is important for the accurate localization and quantification of physiological changes. SDS values as short as 2.15~mm has been proposed to obtain measurements specific to extracerebral physiology in neonates~\cite{brigadoi2015short}. However, in the present simulations, the combined source and detector radii of 3.06~mm define the lower bound for achievable SDSs. 

Modeling of the prism terminals and the limited numerical aperture (NA) of optical fiber bundles lowers the sensitivity in the most superficial voxels and results in a slight increase in the relative sensitivity to the brain versus extracerebral tissue compared to isotropic detection at short SDSs approximately less than 15\,mm. This difference was observed to stem mainly from the restriction in NA, and could be more distinguishable if the surface had a lower effective scattering coefficient and the exit directions would be less random. The resulting reduction in SNR is an important factor to consider. A small NA can be beneficial in reducing temporal dispersion in TD systems and in minimizing changes in TOF due to bending of the fibers, which may otherwise reduce calibration accuracy and affect resulting estimates of optical parameters. In functional imaging of brain activation, however, minimizing photon shot noise is usually considered to be of higher priority. 

The prism-based fiber terminal model with the current parameters rejects a large fraction of photon packets (approximately 98\%) exiting the tissue within the region sampled by each detector optode, however, the efficacy of photon collection is largely determined by the diameter and numerical aperture of the detector fiber bundle and not significantly affected by the prism terminal in this design. To achieve an acceptable SNR for the simulated measurables and Jacobians with less computation, the detector model could be applied only to short and intermediate source--detector separations ($<$ 2\,cm), while the isotropic model is used for larger separations. Detection efficiency could be improved by increasing the fiber bundle diameter, the NA of the fibers, or both. However, increasing the bundle diameter would also increase cable weight and cost. The current prism size can accommodate fibers with a diameter of 2~mm, but increasing the NA to 0.55 would require a larger prism to avoid photon loss. Both improvements would potentially increase the thickness of the probe, which would be problematic in multimodal imaging, for example in simultaneous OT and MEG using superconducting quantum interference devices (SQUID-MEG) or simultaneous OT and transcranial magnetic stimulation (TMS). A heavier probe might also reduce subject comfort, particularly in the case of infants.

\section{Conclusion} 
\label{sec:conclusion}
In this work, pMC-based Jacobian computation was expanded to accommodate FD and TD measurements, especially sensitivities with respect to scattering coefficients, in voxel-based and mesh-voxel hybrid models. The computation was implemented in the open source, GPU-accelerated MCX simulator. A post-processing function was introduced to model detection in prism-terminated optical fiber-based probes. The pMC-derived Jacobians showed excellent qualitative and quantitative agreement with DA-based Jacobians in the diffusive regimes of two segmented neonatal head models. A limitation of the pMC scattering Jacobians was the required large number ($10^{12}$--$10^{13}$) of simulated photon packets to convergence for channels with source--detector separations of 10--30\,mm in the considered head models. The more realistic detector model with numerical aperture limitation marginally reduced sensitivity to most superficial voxels and increased sensitivity in brain tissue at short source--detector separations compared to isotropic reception. Future work addresses novel approaches to utilize the scattering Jacobians in OT image reconstruction.

\section*{Disclosures}
The authors declare that there are no financial interests, commercial affiliations, or other potential conflicts of interest that could have influenced the objectivity of this research or the writing of this paper.

\section*{Code and Data Availability}
A demonstration script for computing all Jacobian types in MCX is available at \url{https://github.com/fangq/mcx/blob/master/mcxlab/examples/demo_replay_all_jacobian.m}. The new Jacobian types are included in the v2025.10 release of MCX. However, for the split-voxel mode or for simulations with very large numbers of detected photon packets, the more recent nightly build available at \url{https://mcx.space/nightly/} is recommended. The function implementing NA-limiting detection is specific to the instrumentation at Aalto University and is therefore not made publicly available, but can be obtained from the authors upon reasonable request. The neonatal head model database used in this study is available at \url{https://www.ucl.ac.uk/dot-hub}.
 
\section*{Acknowledgements}
PH, JO, and IN were supported by the Research Council of Finland (Flagship of Advanced Mathematics for Sensing, Imaging and Modelling, Grant No.~359181). PH was supported by the Alfred Kordelin Foundation. JO was supported by the Finnish Ministry of Education and Culture’s Pilot for Doctoral Programmes (Mathematics of Sensing, Imaging and Modelling) and by the Research Council of Finland (Digital Waters Flagship, Grant No.~359249). IN was supported by the Päivikki and Sakari Sohlberg Foundation. QF was supported by the US National Institutes of Health (NIH) grants R01-GM114365, R01-CA204443, and R01-EB026998. The UCL neonatal models were obtained using data made available from the Developing Human Connectome Project funded by the European Research Council under the European Union's Seventh Framework Programme (FP/2007-2013)/ERC Grant Agreement No.~(319456). 

The authors thank Dr.~Jarno Rantaharju (Aalto University) for assistance with implementing the new Jacobian types to Monte Carlo eXtreme, Dr.~Topi Kuutela (Aalto University) for programming the original version of the scikit-fem-based finite element solver for the diffusion approximation, Profs.~Nuutti Hyvönen and Antti Hannukainen (Aalto University) for proof-reading the manuscript, and Prof.~Hamid Dehghani (University of Birmingham) and Dr.~Niko Hänninen (University of Eastern Finland) for insightful discussions. The computational resources and guidance provided by the Aalto Science-IT project are gratefully acknowledged. A large language model (ChatGPT) and the Overleaf AI writing assistant were used to assist with language and grammar refinement; all generated content was carefully reviewed and edited by the authors. 

\section*{CRediT}
Conceptualization (PH, IN), Data curation (PH, JO), Formal analysis (PH, JO), Funding acquisition (PH, QF, IN), Investigation (PH, JO), Methodology (PH, JO, IN), Resources (QF, IN), Software (PH, JO, QF, IN), Supervision (IN), Validation (PH, JO), Visualization (PH, JO, IN), Writing -- original draft (PH, JO, IN), Writing -- review \& editing (QF). \\

\noindent
\textbf{Pauliina Hirvi}, MSc in Biomedical Engineering, is a PhD candidate in the Inverse Problems group at Aalto University, Department of Mathematics and Systems Analysis, Espoo, Finland. Her research focuses on generating head models and image reconstruction algorithms for optical tomography, and she has 10 years of experience in Monte Carlo eXtreme software.\\

\noindent \textbf{Jaakko Olkkonen}, MSc in Applied Mathematics, is a Doctoral Researcher in the Numerical Analysis group at Aalto University, Department of Mathematics and Systems Analysis, Espoo, Finland, and a Research Scientist in the Autonomous Mapping and Driving group at the Finnish Geospatial Research Institute, Department of Remote Sensing and Photogrammetry, Espoo, Finland. His research focuses on computational modelling of light transport in optical imaging.\\

\noindent \textbf{Qianqian Fang}, PhD, is an associate professor in the Bioengineering Department at Northeastern University, Boston, Massachusetts, United States. Previously, he served as assistant professor of radiology at the Massachusetts General Hospital. Dr.\ Fang is the main developer of the Monte Carlo eXtreme software. His research interests also include translational medical imaging devices, image reconstruction algorithms, scientific data sharing, and high-performance computing tools to facilitate the development of next-generation imaging platforms.\\

\noindent \textbf{Ilkka Nissilä}, DSc, is a staff scientist and leads the Near-infrared Spectroscopy and Imaging group at Aalto University, Department of Neuroscience and Biomedical Engineering, Espoo, Finland. Dr.\ Nissilä led the design and implementation of the FD-DOT instrument at Aalto University (formerly Helsinki University of Technology) and has extensive experience with the application of DOT to neuroimaging in children and adults.

\appendix 

\section{Proof of Theorem \ref{thm:Fréchet_derivatives}}
\label{sec:proof}

\begin{proof}[Proof of Theorem \ref{thm:Fréchet_derivatives}]
Let $\phi(\mu,\lambda) \in H^1(\Omega;\mathbb{C})$ denote the unique weak solution of
\begin{align}
\begin{cases}
\begin{alignedat}{2}
- \nabla \boldsymbol{\cdot} \left( \lambda \nabla \phi\right) 
  + \Big(\mu - \mathrm{i}\frac{\omega}{c}\Big) \phi
  &\;= 0, & &\quad \text{in } \Omega, \\
\frac{1}{4}(1-\rho)\phi+ \frac{1}{2}(1+\rho)\, \boldsymbol{\nu}\boldsymbol{\cdot}\lambda \nabla \phi
  &\;= Q_k,& & \quad \text{on } \partial \Omega,
\end{alignedat}
\end{cases}
\end{align}
for $\mu, \lambda \in L^\infty_+(\Omega; \mathbb{R})$. It is well-known (see, e.g.,~\cite{natterer2002frechet}) that the parameter-to-solution map 
\[
[\!L_+^\infty(\Omega; \mathbb{R})]^2 \ni (\mu, \lambda) \longmapsto \phi(\mu, \lambda) \in H^1(\Omega;\mathbb{C})
\] 
is Fr\'echet differentiable, and the derivative at $(\mu, \lambda) \in [L^\infty_+(\Omega; \mathbb{R})]^2$ in the direction $(\zeta, \vartheta) \in [L^\infty(\Omega; \mathbb{R})]^2$ is the unique $\phi' := D\phi(\mu, \lambda)[(\zeta, \vartheta)] \in H^1(\Omega;\mathbb{C})$ satisfying
\begin{align}
&\int_\Omega \left( \lambda \nabla \phi' \boldsymbol{\cdot} \nabla \overline{\psi} + \Big(\mu -\mathrm{i}\frac{\omega}{c}\Big)\phi' \overline{\psi} \right)\,\mathrm{d}x 
+ \frac{1-\rho}{2(1+\rho)} \int_{\partial \Omega} \phi' \overline{\psi} \,\mathrm{d}S \notag  \\
& \quad = - \int_\Omega \zeta \, \phi(\mu, \lambda) \overline{\psi} \, \mathrm{d}x
- \int_\Omega \vartheta \, \nabla \phi(\mu, \lambda) \boldsymbol{\cdot} \nabla \overline{\psi} \, \mathrm{d}x
\quad \forall \psi \in H^1(\Omega;\mathbb{C}).
\label{eq:forward_derivative_refined}
\end{align}

The trace theorem guarantees that the boundary measurement
\[
M(\mu, \lambda) = \frac{1-\rho}{2(1+\rho)} \int_{\partial \Omega} P_j \, \phi(\mu, \lambda) \, \mathrm{d}S
\] 
is well-defined. Since the mapping $H^1(\Omega; \mathbb{C}) \ni \phi \mapsto \int_{\partial \Omega} P_j \phi \, \mathrm{d}S$ is linear and continuous (by the trace theorem), it is straightforward to deduce that
\begin{align}
DM(\mu, \lambda)[(\zeta, \vartheta)] 
&= \frac{1-\rho}{2(1+\rho)} \int_{\partial \Omega} P_j \, \phi' \, \mathrm{d}S,
\label{eq:measurement_derivative_refined}
\end{align}
where $\phi'$ is the solution of \eqref{eq:forward_derivative_refined}.

Let $\phi^* = \phi^*(\mu, \lambda) \in H^1(\Omega;\mathbb{C})$ denote the unique solution of the weak adjoint problem
\begin{align}
&\int_\Omega \left(\lambda \nabla \phi^* \boldsymbol{\cdot} \nabla \overline{\psi} + \Big(\mu - \mathrm{i}\frac{\omega}{c}\Big) \phi^* \overline{\psi}\right) \, \mathrm{d}x
+ \frac{1-\rho}{2(1+\rho)} \int_{\partial \Omega} \phi^* \overline{\psi} \, \mathrm{d}S \notag \\
& \quad = \frac{2}{1+\rho} \int_{\partial \Omega} P_j \overline{\psi} \, \mathrm{d}S
\quad \forall \psi \in H^1(\Omega;\mathbb{C}).
\label{eq:weak_adjoint_refined}
\end{align}
Choosing $\psi = \overline{\phi'}$ in \eqref{eq:weak_adjoint_refined} and combining with \eqref{eq:forward_derivative_refined} and \eqref{eq:measurement_derivative_refined} yields
\[
DM(\mu, \lambda)[(\zeta, \vartheta)] 
= - \frac{1}{4}(1-\rho) \left( \int_\Omega \zeta \, \phi \, \phi^* \, \mathrm{d}x 
+ \int_\Omega \vartheta \, \nabla \phi \boldsymbol{\cdot} \nabla \phi^* \, \mathrm{d}x \right).
\]

Setting $\mu = \mu_a$ and $\lambda = \kappa(\mu_a, \mu_s)$ results in $\phi(\mu, \lambda) = \Phi_k(\mu_a, \mu_s)$, $M(\mu, \lambda) = \mathcal{M}_{jk}(\mu_a, \mu_s)$, and $\phi^*(\mu, \lambda) = \Phi_j^*(\mu_a, \mu_s)$. Since the Fr\'echet derivative of $\kappa$ at $(\mu_a, \mu_s) \in [L^\infty_+(\Omega; \mathbb{R})]^2$ in the direction $(\zeta, \vartheta)\in [L^\infty(\Omega; \mathbb{R})]^2$ is
\[
D\kappa(\mu_a, \mu_s)[(\zeta, \vartheta)] = -3 \, \kappa(\mu_a, \mu_s)^2 \, \zeta - 3 \, \kappa(\mu_a, \mu_s)^2 (1-g) \, \vartheta,
\]
application of the chain rule yields
\begin{align*}
D\mathcal{M}_{jk}(\mu_a, \mu_s)[(\zeta, \vartheta)] 
&= \frac{1}{4}(1-\rho) \int_\Omega \zeta \left( 3 \, \kappa^2 \, \nabla \Phi_k \boldsymbol{\cdot} \nabla \Phi_j^* - \Phi_k \Phi_j^* \right) \, \mathrm{d}x \\
&\quad + \frac{3}{4}(1-\rho) \int_\Omega \vartheta \, \kappa^2 (1-g) \, \nabla \Phi_k \boldsymbol{\cdot} \nabla \Phi_j^* \, \mathrm{d}x.
\end{align*}
\end{proof}

\bibliographystyle{unsrt}
\bibliography{References.bib}
\end{document}